\documentclass[a4paper,twocolumn,11pt]{quantumarticle}
\pdfoutput=1
\usepackage[utf8]{inputenc}
\usepackage[english]{babel}
\usepackage[T1]{fontenc}
\usepackage{amsmath}
\usepackage{hyperref}

\usepackage{tikz}
\usepackage{lipsum}
\usepackage{adjustbox}
\usepackage[numbers,sort&compress]{natbib}
\usepackage{amssymb}

\usepackage{braket}
\usepackage{natbib}
\usepackage{color}
\usepackage{url}
\usepackage{graphicx}
\usepackage{multirow}
\usepackage{amsthm}
\usepackage[mode=buildnew]{standalone}
\usepackage{algorithm2e}
\newtheorem{theorem}{Theorem}
\newtheorem{corollary}{Corollary}[theorem]
\newtheorem{lemma}{Lemma}
\newtheorem{definition}{Definition}

\newtheorem{proposition}{Proposition}

\newcommand{\class}[1]{\textup{\textbf{#1}}}

\newcommand{\Poly}{\class{P}}
\newcommand{\NP}{\class{NP}}
\newcommand{\coNP}{\class{coNP}}
\newcommand{\Pitwo}{\mathbf{\Pi_2^p}}
\newcommand{\Sigmatwo}{\mathbf{\Sigma_2^p}}
\newcommand{\diag}{\text{diag}}
\newcommand{\problem}[1]{\textup{\textsf{#1}}}

\newcommand{\GlobalStoq}{\problem{GlobalStoq}}
\newcommand{\forallexistSAT}{\mathsf{\forall \exists}\text{-}$3$ \problem{-SAT}}
\newcommand{\minmaxSAT}{\problem{MINMAX}-$2$-\problem{SAT}}
\newcommand{\poly}{\text{poly}}

\newcommand{\term}{{\sf TermStoq}}
\newcommand{\glob}{{\sf GlobalStoq}}
\newcommand{\lh}{{\sf LH}}

\newcommand{\norm}[1]{||#1||}

\begin{document}

\title{Termwise versus globally stoquastic local Hamiltonians: questions of complexity and sign-curing}

\author{Marios Ioannou}
\affiliation{Dahlem Center for Complex Quantum Systems, Freie Universität Berlin, 14195 Berlin, Germany}
\email{marios.ioannou@fu-berlin.de}
\author{Stephen Piddock}
\affiliation{School of Mathematics, University of Bristol, Bristol, UK}
\affiliation{Heilbronn Institute for Mathematical Research, Bristol, UK}

\author{Milad Marvian}
\affiliation{Department of Electrical and Computer Engineering and Center for Quantum Information and Control (CQuIC), University of New Mexico, Albuquerque, NM 87131, United States of America }

\author{Joel Klassen}
\affiliation{Phasecraft Ltd.}

\author{Barbara M. Terhal}
\affiliation{QuTech, Delft University of Technology, P.O. Box 5046, 2600 GA Delft, The Netherlands}
\affiliation{
JARA Institute for Quantum Information, Forschungszentrum Juelich, D-52425 Juelich, Germany}

\maketitle


\begin{abstract}
	We elucidate the distinction between global and termwise stoquasticity for local Hamiltonians and prove several complexity results. We show that the stoquastic local Hamiltonian problem is $\textbf{StoqMA}$-complete even for globally stoquastic Hamiltonians. We study the complexity of deciding whether a local Hamiltonian is globally stoquastic or not. In particular, we prove \coNP{}-hardness of deciding global stoquasticity in a fixed basis and $\Sigmatwo$-hardness of deciding global stoquasticity under single-qubit transformations. As a last result, we expand the class of sign-curing transformations by showing how Clifford transformations can sign-cure a class of disordered 1D XYZ Hamiltonians.
\end{abstract}

\maketitle


\section{ Introduction}

In classical computational physics the simulation of quantum systems often uses numerical methods based on Monte Carlo algorithms. The goal of such algorithms is to sample from a probability distribution of interest, for example the Gibbs distribution associated with a Hamiltonian. One of the fundamental problems that can arise in such simulations goes under the name of the "sign problem" and is due to intrinsic differences in complexity between classical and quantum systems.

Loosely speaking, a Hamiltonian is considered "sign problem free" when it satisfies criteria which render it amenable to Monte Carlo simulations. 
It should be immediately noted however, that being sign problem free is not a sufficient criterion 
for efficient classical simulability. This can be seen from the fact that there exist physical systems (including such that correspond to purely classical Hamiltonians) which encode \NP{}-hard problems. 
 
 Understanding a minimal set of conditions which guarantees a Hamiltonian to be sign problem free or realising ways in which the sign problem can be "cured" for certain Hamiltonians or even simply identifying systems which have an inherent sign problem are questions that are of high relevance and which have been studied extensively in computational physics and quantum chemistry over many years, see e.g. \cite{Huffman,Chandrasekharan, Aarts, Wessel}. 

More recently, the emergence of quantum computing has led to the formal study of the complexity of physically motivated computational problems, such as the problem of estimating the ground state energy of a Hamiltonian. To capture the class of Hamiltonians which do not suffer from the sign problem and study their specific computational power, the notion of so-called \emph{stoquastic} Hamiltonians was introduced in \cite{BDT:stoq}. 
In particular, a Hamiltonian is called stoquastic with respect to the computational basis $\{\ket{x}\}$ when $H$ is real and additionally $\bra{x} H \ket{y} \leq 0, \forall x\neq y$. Stoquasticity of $H$ ensures that the thermal state $e^{-\beta H}$ is a non-negative matrix \cite{AG:NP} in the computational basis for any $\beta \geq 0$, and that the ground state has non-negative amplitudes. In the setting of path integral quantum Monte Carlo, having no sign problem translates to the partition function $Z(\beta)={\rm Tr} (e^{-\beta H})$ being expressible as a sum over non-negative weights. Therefore every stoquastic Hamiltonian is also sign problem free, however, it is worth mentioning that stoquasticity is not a necessary condition for avoidimg the sign problem \cite{Suzuki, GAH:QMC,GH:sign}.
For a detailed exposition of the connection between stoquastic Hamiltonians and the sign problem in the specific setting of the path integral quantum Monte Carlo method we refer to Appendix \ref{sec:signProblem}. 

Research on the power of stoquastic Hamiltonians has gone into a number of different directions. Of particular interest is the study of stoquastic adiabatic computation, with quantum annealing being a prime example. While it has been shown that adiabatic computing using stoquastic frustration-free Hamiltonians can be efficiently classically simulated \cite{BravyiTerhalMA}, a recent breakthrough result shows that stoquastic adiabatic computation can be subexponentially more powerful than classical computation with respect to the number of queries to an oracle \cite{hastings:stoq-adia, GV:nosign}. On the other hand, as it is strongly believed that general Hamiltonians are computationally more powerful than stoquastic Hamiltonians, there has been an increased interest in engineering non-stoquastic Hamiltonians for adiabatic computing \cite{Ozfidan2019}. Interestingly, recent evidence \cite{CAHY:design} suggests that the run-time for quantum adiabatic optimization algorithms might not benefit from using non-stoquastic over stoquastic Hamiltonians.

As mentioned above, in \cite{BDT:stoq} a Hamiltonian is called stoquastic when its matrix representation consists of real entries and the off diagonal elements are non-positive. In this work we will refer to such Hamiltonians as being \emph{globally} stoquastic. However, in much of the literature concerning stoquasticity, e.g. \cite{AG:NP}, \cite{BBT:stoq}, \cite{BravyiTerhalMA}, the stronger notion of \emph{termwise} stoquastic local Hamiltonians has been used. A local Hamiltonian is said to be termwise stoquastic with respect to the computational basis if there exists a decomposition into local terms each of which itself is stoquastic, see Definitions \ref{def:globalStoq} and \ref{def:termStoq}. 
In fact these two classes, termwise and globally stoquastic, do not strictly coincide, as we shall illustrate. Given that globally stoquastic Hamiltonians are a strictly larger class, it is natural to ask which complexity results established for termwise stoquastic Hamiltonians may be extended to globally stoquastic Hamiltonians.
In this work we study this in the context of the local Hamiltonian problem showing that the complexity for global and termwise stoquastic Hamiltonians coincides.

In the stoquastic local Hamiltonian problem one is promised that the Hamiltonians are stoquastic in the computational basis. However, in practice given an unknown Hamiltonian one would ideally like to check whether this is indeed the case. This is a point that has been largely taken for granted, and something we address in this work. For the simplest case of deciding if a $k$-local Hamiltonian is $k$-termwise stoquastic in a fixed basis, an efficient classical algorithm has been proposed~\cite{Marvian2018}. However, we show that for globally stoquastic Hamiltonians this is generically a computationally-hard problem. 

Furthermore, as stoquasticity is a basis dependent notion, one might hope to find a suitable basis in which to express the Hamiltonian such that it becomes sign problem free if it isn't in the computational basis. In Refs.~\cite{Marvian2018, KT:stoq} and \cite{Klassen2019} the authors examined whether a local Hamiltonian can be computationally-\emph{termwise}-stoquastic, that is, whether the sign problem can be cured by a computationally-efficient basis changing transformation. Examples of such transformations are single-qubit (product) unitary gates or single-qubit Clifford transformations. In a similar vein, one can ask whether the sign problem can approximately be cured \cite{Hangleiter2019} or, whether a low-depth quantum circuit can transform away the sign problem in the ground state \cite{torlai}.
In these previous works it was found that deciding whether a local Hamiltonian is termwise stoquastic or not, is computationally-hard \cite{Marvian2018,Klassen2019}. In particular, in Ref.~\cite{Klassen2019} it was shown that deciding if a two-local Hamiltonian is termwise stoquastic under single-qubit unitary transformations, is \NP-hard, while Ref.~\cite{Klassen2019} also gave an efficient algorithm for deciding termwise stoquasticity under single-qubit unitary transformations for two-local Hamiltonians {\em without} any 1-local terms. Naturally, these hardness results do not preclude the existence of heuristic or approximate strategies which find basis changes which can reduce the severity of the sign problem as explored in \cite{Hangleiter2019}, but this has so far not been explored systematically.

As mentioned, in these previous works the problem of sign-curing via local transformations was studied with the notion of termwise stoquasticity in mind. However, it is not clear if the results would be identical if we study the problem in the setting of global stoquasticity. Our paper also answers this question showing that the complexity for deciding global stoquasticity is fundamentally different.

We start this paper by clarifying the distinction between the notion of global stoquasticity versus termwise stoquasticity for a local Hamiltonian (Section \ref{global_vs_termwise}). Although the two definitions coincide for (1) 2-local Hamiltonians and (2) Hamiltonians where each qubit interacts with $O(1)$ others, there are 3-local Hamiltonians which are globally stoquastic but not termwise stoquastic.

In Section \ref{sec:globLH} we consider the complexity of the local Hamiltonian problem, which is known to be \textbf{StoqMA}-complete for termwise stoquastic Hamiltonians \cite{BBT:stoq}. We show that the local Hamiltonian problem is contained in \textbf{StoqMA} for globally stoquastic Hamiltonians as well (Theorem~\ref{thm:globalLH}) and therefore that the complexity is independent of the notion of stoquasticity used. Another problem that has been studied extensively is the so called quantum $k$-SAT problem. The stoquastic version of this problem amounts to deciding whether a $O(1)$-termwise stoquastic Hamiltonian is frustration-free. It was shown that stoquastic $k$-SAT is contained in \textbf{MA} \cite{BravyiTerhalMA}, and, assuming a constant spectral gap, in \NP{} \cite{AG:NP}. The notion of globally stoquastic Hamiltonians suggests a slight generalization of stoquastic $k$-SAT which we also believe to be contained in \textbf{MA}.

In Sections~\ref{sec:fixedbasis}, \ref{sec:stoq-lb} and \ref{sec:Clifford} we switch our attention to the complexity of deciding whether a Hamiltonian is termwise or globally stoquastic and whether the Hamiltonian can made stoquastic by local basis changes or Clifford transformations.
We first show that while it is computationally tractable to decide whether an $n$-qubit Hamiltonian is termwise stoquastic \cite{Marvian2018}, it is \coNP{}-complete to decide whether it is globally stoquastic (Theorem \ref{thrm:Globalstoqu_conp}) even in a fixed basis. A definition of stoquasticity which cannot be easily checked may seem impractical, however we also show that for most if not all physically relevant Hamiltonians the general hardness obstruction does not apply. 
Next, we consider the problem of sign curing under transformations of the form $\bigotimes_{i=1}^n U_i$ where $U_i \in  {\sf U(2)}$ to show that deciding global stoquasticity under such local unitary transformations is $\Sigmatwo$-hard (Section \ref{sec:stoq-lb}). These results fit into the existing picture as shown in Table~\ref{hardnessResults}. 
As we mentioned earlier, if the Hamiltonian is 2-local or every qubit interacts with only $O(1)$ other qubits, then global stoquasticity is equivalent to termwise stoquasticity. The Hamiltonians in our hardness constructions minimally violate these conditions: they are 3-local and there is a single qubit which can interact with all other qubits.

Taking these results together one arrives at the curious observation that, since \[\textbf{StoqMA} \subseteq \textbf{AM} \subseteq \Pitwo=\text{co}\Sigmatwo,\]
 the complexity of determining if a local Hamiltonian is not globally stoquastic modulo a local basis change, is strictly harder than approximating its ground state energy, unless all of the above complexity classes are equal. While single qubit product unitaries form an important class of computationally efficient basis transformations, we would ideally like to extend our results to other classes of transformations such as shallow depth circuits or Clifford circuits. It is intuitively clear that different basis transformations will sign-cure different classes of Hamiltonians. 

In our last result, we give a family of disordered 1D XYZ models which can be sign-cured by Clifford transformations while it is not possible to sign-cure using single qubit product unitary transformations (Theorem \ref{clifford2}). This makes the above intuition rigorous and extends the class of Hamiltonians which are, modulo a computationally efficient basis transformation, stoquastic.

  For the technical sections of this paper familiarity with the complexity classes \NP{}, \Poly{} and \coNP{} will be assumed. We will also denote a general polynomial in $n$ as ${\rm poly}(n)$ when we don't care about its degree or prefactors specifically.

\begin{table*}[htb]
\centering
	\begin{tabular}{|c|c|c|}
		\hline
		Stoquastic & $O(1)$-Termwise & Global  \\
		\hline
		Fixed Basis & $ {
			\bf P}$ \cite{Marvian2018} &\coNP{}-complete [*]\\
		1-Local Unitary & \NP{}-hard \cite{Klassen2019}  & $\Sigmatwo$-hard [*] \\
		\hline
	\end{tabular}
	\caption{Hardness results for deciding membership of a $O(1)$-local Hamiltonian to a particular stoquasticity class. [*] indicates proof given in this work.
	The complexities on the bottom row of the table are \NP-complete and $\Sigmatwo$-complete respectively when there is the additional promise that in the YES case the curing unitary $U$ can be described in such a way that allows for efficient computation of the matrix entries of $UHU^{\dagger}$, see Theorem \ref{thm:Sigmatwo}.}
	\label{hardnessResults}
\end{table*}

\section{Global versus termwise stoquasticity}\label{global_vs_termwise}

There are currently two definitions of stoquasticity which are used interchangeably in the literature. We will here make the distinction between them and study their differences. We will be concerned with $n$-qubit $k$-local Hamiltonians $H$: such Hamiltonians can be written as a sum of Hermitian terms each of which acts on at most $k$ qubits nontrivially and $k$ is usually $O(1)$\footnote{For sparse Hamiltonians without locality structure, for example the Laplacian on a graph \cite{CAHY:design}, the notion of termwise stoquasticity is not meaningful.}. We start by formally defining what we referred to in the introduction as global stoquasticity,

\begin{definition}[Globally Stoquastic]\label{def:globalStoq}
	A $n$-qubit $k$-local Hamiltonian $H$ is globally stoquastic with respect to a basis $\{\ket{x}\}$, if $H$ is real and obeys $\forall x\neq y,\;\bra{x} H \ket{y} \leq 0$.
\end{definition}

One way to further restrict the class of stoquastic Hamiltonians is by asking whether the local Hamiltonian is decomposable as a sum of local terms, each of which itself is stoquastic. More precisely, one can define the class of local Hamiltonians which are termwise stoquastic:

\begin{definition}[Termwise Stoquastic]\label{def:termStoq}
	A $n$-qubit $k$-local Hamiltonian $H$ is $m$-termwise stoquastic with respect to a basis $\{\ket{x}\}$, if it admits a decomposition into $m$-local terms $H=\sum_{i=1}^I D_i$ such that each $m$-local Hermitian term $D_i$ is real and obeys $\forall x\neq y,\;\bra{x} D_i \ket{y} \leq 0$.
\end{definition}

{\em Remark}: For the set of real $k$-local Hamiltonians, we can choose a real Pauli basis with each basis element acting on at most $k$ qubits nontrivially. 
We can then view a $k$-local Hamiltonian with $k=O(1)$ as a point in $ \mathbb{R}^d$ with $d \leq \sum_{k'=0}^k {\binom{n}{k'}}(4^{k'}-1)={\rm poly}(n)$. 
On the other hand, the set of $m$-local stoquastic matrices $D$ generate a convex cone ${\cal C}_n(m)$. 
The question whether a $k$-local Hamiltonian $H$ is $m$-termwise stoquastic is thus the question whether $H$ lies in ${\cal C}_n(m)$. Caratheodory's theorem then tells us that if $H$ lies in the cone, it is supported on $I \leq d+1\leq {\rm poly}(n)$ points $D_i \in {\cal C}_n(m)$.\\

Note that the definition for global stoquasticity is equivalent to being $(m=n)$-termwise stoquastic and being $m$-termwise stoquastic implies being $(m+1)$-termwise stoquastic. One can further see that any $m$-termwise stoquastic Hamiltonian is also globally stoquastic. 
As was stated in \cite{BDT:stoq} and later shown explicitly in \cite{Klassen2019} the converse also holds for two-local Hamiltonians, i.e. every two-local globally stoquastic Hamiltonian is $2$-termwise stoquastic. Such equivalence can be generalized:

\begin{lemma}
	If each qubit of a $k$-local Hamiltonian $H$ interacts with at most $l=O(1)$ other qubits and $k=O(1)$, then global stoquasticity implies $m$-termwise stoquasticity, where $m=kl=O(1)$.
	\label{lem:bd}
\end{lemma}
\begin{proof}
	Assume global stoquasticity of the Hamiltonian. Write $H=\sum_S M^{(S)}$ where $M$ are all terms in $H$ which flip the qubits in the subset $S$ with $|S|\leq k$. When $S=\emptyset$, $M^{(\emptyset)}$ collects all the diagonal terms. Global stoquasticity implies that each term $M^{(S)}$ is stoquastic as each such term gives rise to distinct non-zero matrix elements $\bra{x} H \ket{y}$. Terms in $M^{(S)}$ for $S \neq \emptyset$ involve qubits in the subset $S$ and at most $|S|\times l$ qubits outside of $S$ (where they act Z-like). Hence $M^{(S)}$ is at most $kl$-local and $H$ is $m$-termwise stoquastic with $m\leq kl$.
\end{proof}

The open question is thus whether the notion of global stoquasticity and $O(1)$-termwise stoquasticity is the same for arbitrary $k$-local interactions. One can in fact construct the follow 3-local counterexample, due to Bravyi \cite{Bravyi_private}.


\begin{proposition}\label{prop1} 
	There exists a class of $3$-local globally-stoquastic Hamiltonians which are not 3-termwise stoquastic.
\end{proposition}

\begin{proof}
	Consider a graph $G=(V,E)$ with $n-1$ vertices and define an $n-1$  Ising Hamiltonian on it: 
	\begin{equation}
	H_{\text{Ising}}=\sum_{(i,j)\in E}J_{ij}Z_iZ_j,
	\end{equation}
	where $J_{ij}\in \{-1,+1\}$. Now consider the Hamiltonian $H=X_0 \otimes(E_0-H_{\text{Ising}})$ where $E_0=\min_{x\in \{0,1\}^{n-1}}[\bra{x}H_{\text{Ising}}\ket{x}]< 0$ is the lowest eigenvalue of $H_{\text{Ising}}$. Note that here as well as in later parts we have abused notation for the sake of readability, by writing $E_0$ instead of $(E_0) I$ i.e the identity matrix weighted by $E_0$.
	We show that $H$ is globally-stoquastic but cannot be written as a sum of stoquastic 3-local terms when the Ising Hamiltonian is frustrated (and therefore $E_0$ is larger than the sum of the minimal values of each term).  
	The only non-zero off-diagonal matrix elements of $H$ are $\bra{x_0,x} H \ket{y_0,x}=\bra{x} (E_0-H_{\rm Ising}) \ket{x}\leq 0$ where $x \in \{0,1\}^{n-1}$ and $x_0 \neq y_0$. Since $E_0$ is the smallest eigenvalue of $H_{\rm Ising}$, it follows immediately that $H$ is globally stoquastic.
	
	Now consider whether $H$ can be $m$-termwise stoquastic. The only way for $H$ to decompose into $m$-termwise stoquastic terms is for it to take the form $H=X_0 \otimes \sum_{i,j} D_{i,j}$ where each $X_0 \otimes D_{i,j}$ is stoquastic and $D_{i,j}$ acts on qubits at vertices $i$ and $j$ and is furthermore diagonal. It is straightforward to see that one does not need to consider more general decompositions of the form $H=\sum_{i,j,k}D_{i,j,k}$ as along with any such term $D_{i,j,k}$ which cannot be written as $X_0 \otimes D_{i,j}$ there would have to also be a term $-D_{i,j,k}$. Therefore for off-diagonal terms $D_{i,j,k}$, either of the two will be non stoquastic. 
	Each term $X_0 \otimes D_{i,j}$ is stoquastic iff $\forall x,\; \bra{x} D_{i,j} \ket{x} \leq 0$ and the most general form for $D_{i,j}$ is
	\begin{equation}
	D_{i,j}=\alpha_{ij}^{I}I_{ij}+\alpha_{ij}^{IZ}Z_j+\alpha_{ij}^{ZI}Z_i+J_{ij}Z_iZ_j,
	\end{equation}
	with the restrictions $\sum_{i,j:i<j}\alpha_{ij}^{I}=E_0$, $\sum_{i:i<j}\alpha_{ij}^{IZ}=0, \; \forall j$ and $\sum_{j:j>i}\alpha_{ij}^{ZI}=0,\; \forall i$.
	
	The condition $\forall x,\; \bra{x} D_{i,j} \ket{x} \leq 0$ implies that
	
	\begin{align}
	    \alpha_{ij}^I &\leq \min_{\Delta_i=\pm 1,\Delta_j=\pm 1}[-\Delta_j\alpha_{ij}^{IZ}-\Delta_i\alpha_{ij}^{ZI}-\Delta_i\Delta_j J_{ij}] \nonumber \\ &\leq -1.
	\end{align}
	
	If $H_{\text{Ising}}$ is frustrated, $E_0 >  -\sum_{(i,j)\in E}$ as not all terms contribute negatively: this is then in contradiction with the requirement that $\alpha_{ij}^I \leq -1$ and $\sum_{i,j:i< j}\alpha_{ij}^I=E_0$.
\end{proof}


Note that the result does not preclude the possibility that Hamiltonians constructed in this proposition are $(k\geq 4)$-termwise stoquastic. In fact, if we consider a slightly more general setting we will see that by increasing the locality of the terms into which $H$ is decomposed, $H$ can be seen to be $(m>k)$-termwise stoquasic even if it is not $k$-termwise stoquastic. This is due to the correspondence between the frustration of the classical Hamiltonian acting on $n-1$ qubits and the locality of the terms necessary to have a termwise decomposition. As we shall see in Section \ref{sec:hardglob} this correspondence also leads to the hardness of deciding global stoquasticity. Building up on the idea of Proposition \ref{prop1} we first show that for Hamiltonians of this form, a termwise stoquastic decomposition exists iff a frustration free decomposition exists for the classical Hamiltonian acting on $n-1$ qubits. 
\begin{proposition}
	Let $H_{\rm class}$ be a traceless $k$-local Hamiltonian on $n-1$ qubits, with terms diagonal in the computational basis and let $E_0 < 0$ be its lowest eigenvalue. Let $H=X_0 \otimes (E_0-H_{\rm class})$.  $H$ is $(m+1)$-termwise stoquastic if and only if there exists a $(m\geq k)$-local decomposition of $H_{\rm class}$ which is frustration-free. The frustration-free $m$-local decomposition is a decomposition $H_{\rm class}=\sum_i D_i$ where each Hermitian classical term $D_i$ acts on at most $m$ qubits non-trivially and $E_0=\sum_i E_g(D_i)$ with $E_g(D_i)$ the lowest eigenvalue of the term $D_i$.
	\label{prop:class}
\end{proposition}

\begin{proof}
	Assume we have a frustration-free $m$-local decomposition $H_{\rm class}=\sum_i D_i$. A term $X_0 \otimes (E_0(D_i)-D_i)$ is then stoquastic and $(m+1)$-local and hence the sum of these terms is a $(m+1)$-termwise stoquastic decomposition of $H$. Assume $H$ is $(m+1)$-termwise stoquastic, this then induces (by the same argument as in the proof of proposition \ref{prop1}) a decomposition of $H$ of the form $H=\sum_i X_0 \otimes (\alpha_i I -\tilde{D}_i)$ where $\tilde{D}_i$ is traceless, diagonal and $m$-local, and stoquasticity of each term implies that $\alpha_i -E_0(\tilde{D}_i)\leq 0$. In addition, we must have $\sum_i \alpha_i=E_0$ and 
	$\sum_i \tilde{D}_i=H_{\rm class}$ implying that $E_0 \leq \sum_i E_0(\tilde{D}_i)$ which is only possible when the $\tilde{D}_i$ form a frustration-free decomposition of $H_{\rm class}$.
\end{proof}

In anticipation of our later result, Theorem \ref{thrm:Globalstoqu_conp}, we note that we can already see this proposition as a potential obstacle to efficiently deciding whether a Hamiltonian is globally stoquastic. In particular, one strategy for deciding whether the Hamiltonian is globally stoquastic would go via finding an $m$-termwise stoquastic decomposition. However, being able to do so would imply the ability to find an $(m-1)$-local frustration-free decomposition for $H_{\rm class}$. Since this would solve the problem of finding the ground state energy for classical Hamiltonians which is generically \NP{}-complete, see Theorem \ref{Barahona}, we conclude that this strategy will not be computationally efficient. In Theorem \ref{thm:coNP} we will indeed prove that the problem of deciding whether a Hamiltonian is globally stoquastic is \coNP{}-complete by showing that a YES instance of an \NP{}-complete classical lowest eigenvalue problem can be 1-1 mapped onto a NO instance of the problem.

\section{Local Hamiltonian problem for stoquastic Hamiltonians}
\label{sec:globLH}

In the previous section we have seen the difference between global and termwise stoquasticity resulting from their definitions. In this section we will give an overview of the complexity theoretic results known for both termwise and globally stoquastic Hamiltonians. In particular we will be looking at the well-studied local Hamiltonian problem which can be seen as a formalisation of the ground energy problem in physics. For completeness we recall the definition of the standard local Hamiltonian problem (\lh{}).

\begin{definition}\label{LH-problem}
    Let $H=\sum_{i=1}^m H_i$ be a Hamiltonian acting on $n$ qubits where $m= \poly (n)$ and for all $i$, $H_i$ has non-trivial support on at most $k$ qubits (i.e $H$ is $k$-local). Furthermore, let the operator norm $||H_i|| < \poly (n)$ for all $i$. Given are also two constants $a,b$ with $|a-b|\geq 1/{\rm poly}(n)$. The $k$-local Hamiltonian problem (\lh{}) amounts to deciding whether 
    \begin{enumerate}
        \item YES: the smallest eigenvalue of $H$ is less than or equal to $a$, 
        \item NO: the smallest eigenvalue of $H$ is greater than $b$,
    \end{enumerate}
    under the promise that one of the two is true.
\end{definition}

It was first shown by Kitaev that this problem \lh{} is \textbf{QMA}-complete. The natural question then arises whether the problem remains hard under further restrictions on the Hamiltonian. Here we will list the known results concerning termwise stoquastic Hamiltonians and provide an extension of the proof of containment in \textbf{StoqMA} found in \cite{BBT:stoq}. With this new result we show that the complexity of the stoquastic local Hamiltonian problem is the same, independent of the notion of stoquasticity used. 

\subsection{Complexity of the globally stoquastic LH problem}

In \cite{BDT:stoq} the complexity of the globally stoquastic local Hamiltonian problem was analysed and found to be \textbf{MA}-hard and furthermore placed in \textbf{AM}.
The complexity classes \textbf{MA}$\subseteq$\textbf{AM} can be seen as probabilistic analogues to \NP{} with respectively one or two rounds of communication with the prover. The crucial ingredient for the proof of  containment in \textbf{AM} lies in the fact that the matrix 
\begin{equation}
    G=\frac{1}{2}\left(I-\frac{H}{||H||}\right), 
\end{equation}
is non-negative which is ensured when $H$ is globally stoquastic.

Under the additional assumption of termwise stoquasticity it was later shown that the Local Hamiltonian problem is \textbf{StoqMA}-complete \cite{BBT:stoq}. \textbf{StoqMA} is a complexity class contained in $\textbf{AM}$ and is defined as \textbf{QMA} but with a restricted verifier circuit. A stoquastic verifier circuit includes ancilla qubits initialised in $\ket{0}$ or $\ket{+}$, classical reversible gates, and single final measurement of a qubit in the $X$ basis. The verifier accepts on measuring $\ket{+}$ and rejects on measuring $\ket{-}$. A promise decision problem is contained in \textbf{StoqMA} when in the YES case, there exists a witness $\ket{\psi}$ which leads to an acceptance probability which is at least $p_{\rm yes}$. In the NO case, for all witnesses, the acceptance probability is at most $p_{\rm no}$ and one has $p_{\rm yes}-p_{\rm no}\geq \frac{1}{{\rm poly}(n)}$. For a precise definition, see Definition 4 in \cite{BBT:stoq}.
What is noteworthy is that one cannot necessarily amplify the gap between $p_{\rm yes}$ and $p_{\rm no}$ in \textbf{StoqMA}, see \cite{AGL}.


Our new contribution is to show that the globally stoquastic local Hamiltonian problem is also contained in \textbf{StoqMA}.

\begin{theorem} \label{thm:globalLH}
    The local Hamiltonian problem is contained in \textbf{StoqMA} for globally stoquastic $O(1)$-local Hamiltonians.
\end{theorem}

Taken together with the result that 2-local termwise stoquastic local Hamiltonian problem is \textbf{StoqMA}-hard \cite{BBT:stoq}, this gives a tight characterisation of the complexity of the globally stoquastic local Hamiltonian problem.
\begin{corollary}\label{stoqma-complete}
 The local Hamiltonian problem is \textbf{StoqMA}-complete for globally stoquastic Hamiltonians.
\end{corollary}

To prove Theorem \ref{thm:globalLH}, we adapt the proof that the termwise stoquastic Local Hamiltonian problem is in \textbf{StoqMA} from \cite{BBT:stoq}. We need two lemmas, the first provides a decomposition of a globally stoquastic Hamiltonian into terms of a particular form, and the second shows that such terms can be effectively measured by a stoquastic verifier circuit.
\begin{lemma}
\label{lem:globaldecomposition}
Let $H=\sum_{i=1}^m h_i$ be a globally stoquastic $k$-local Hamiltonian. Then there exists $\beta <0$ such that 
\begin{equation}
H+ \beta I = - H_0  +\sum_{j= 1}^{m'} U_j (- X \otimes H_j) U_j^{\dagger}
\label{eq:globaldecomposition}
\end{equation}
where $U_j$ is a quantum circuit of $X$ and CNOT gates and $H_j$ is a classical (diagonal) Hamiltonian with $H_j\ge 0$ for $j\geq 0$. The number of terms in the sum $m'\le m 2^{2k}$. One can (classically) efficiently find this decomposition, i.e. determine $\beta,U_j$ and the description of $H_j$ for $j\geq 0$.
\end{lemma}

\begin{proof}
Let $-H_0$ be the diagonal part of $H$, shifted down by $\beta I$ if necessary to ensure that $H_0 \ge 0$.

For a subset $S$ of the qubits, and a bit string $x \in \{0,1\}^{|S|}$, let 
\begin{align*}
H_{S,x}&=-\sum_{y\in \{0,1\}^{n-|S|}}\bra{xy}H\ket{\bar{x}y}\ket{y}\bra{y}, 
\end{align*}
where for notational convenience we have written the qubits of $S$ first.
Note that $H_{S,x}$ is the diagonal part of $-(\bra{x}\otimes I)H (\ket{\bar{x}} \otimes I)$, and therefore 
$H_{S,x}$ is a classical local Hamiltonian (Note that $H_{S,x}$ can be $I$).
For all $y$, $\bra{y}H_{S,x}\ket{y}=-\bra{xy}H\ket{\bar{x}y}\ge 0 $ since $H$ is globally stoquastic.

Let $U_{S,x}$ be a unitary that maps $\ket{x}$ to $\ket{0^{|S|}}$ and $\ket{\bar{x}}$ to $\ket{10^{|S|-1}}$.
Then the off-diagonal part of $H$ is equal to 
\begin{align}
&-\sum_{S,x} \left(\ket{x}\bra{\bar{x}}+\ket{\bar{x}}\bra{x}\right)  \otimes H_{S,x} \\
& =-\sum_{S,x} U_{S,x}\left( X \otimes \ket{0}\bra{0}^{\otimes |S|-1}\right)U_{S,x}^{\dagger} \otimes H_{S,x}
\end{align}
The claim now follows taking $H_j= \ket{0}\bra{0}^{\otimes |S|-1} \otimes H_{S,x}$, and $U_j=U_{S,x}\otimes I$.

We now count how many terms in this decomposition are non-zero.
Each of the $m$ original terms $h_i$ of $H$ acts non-trivially on a set $T_i$ of at most $k$ qubits and so can contribute to  
$H_{S,x}$ for subsets $S \subseteq T_i  $, of which there are at most $2^k$. 
And for each subset $S$, there are at most $2^{|S|}$ terms $H_{S,x}$, resulting in a bound on the total number $m'$ of $H_j$ terms of $m'\le m2^{2k}$. 

For efficiency of the construction, note that 
 $\beta$ can be set to $\beta=-\sum_i \norm{\diag(h_i)}$, where $\diag(h)$ denotes the diagonal part of $h$.
We can go through all $m2^k$ possible subsets $S$ described in the previous paragraph. Given $S$, all $H_{S,x}$ and $U_{S,x}$ can then efficiently be constructed.
\end{proof}

Next we need to show we can use a stoquastic verifier circuit to effectively measure a term of the form $-H$ or $-X \otimes H$, where $H$ is a classical Hamiltonian.
\begin{lemma}
\label{lem:verifiermeasure}
Let $H$ be a classical local Hamiltonian with $H\ge 0$, and with $\norm{H} \le M$. Let $\Pi_{\ge \alpha}$ be the projector onto strings $\ket{z}$ such that $\bra{z}H\ket{z} \ge \alpha$. 
Then there exists a stoquastic verifier circuit such that the acceptance probability on input $\ket{\psi}$ is equal to $\bra{\psi}\frac{I+G}{2}\ket{\psi}$ where $G$ can be:
\begin{enumerate}

    \item $\ket{0}\bra{0}$ or $\ket{0}\bra{0}\otimes X$
    \item $\Pi_{\ge \alpha}$ or $\Pi_{\ge \alpha}\otimes X$
    \item $\frac{1}{M}H$ or $\frac{1}{M}H \otimes X$.
\end{enumerate}
\end{lemma}

\begin{proof}
    1. is proved in Lemma~2 and Lemma~3 of \cite{BBT:stoq}. Here one uses that one can write the final measurement $\ket{+}\bra{+}=\frac{I+X}{2}=W^{\dagger}  \frac{I+\ket{0}\bra{0}}{2} W$ where $W$ is some specific classical circuit using ancilla qubits in $\ket{0}$ and $\ket{+}$ and $\ket{0}\bra{0}$ is the projector onto any of the input qubits. Similarly, \cite{BBT:stoq} proves that$\ket{+}\bra{+}=\frac{I+X}{2}=W^{\dagger}  \frac{I+\ket{0}\bra{0}\otimes X}{2} W$ for some classical circuit $W$.

For 2. the stoquastic verifier circuit can use its ancilla space to compute whether the energy of a given string $\ket{z}$ is above or below the threshold $\alpha$ and store the answer in a single output bit. Thus the measurements of 2. can be implemented using 1.

The measurements of 3. can be implemented by picking $\alpha$ uniformly at random from $[0,M]$ and then measuring $\Pi_{\ge \alpha}$ (or $\Pi_{\ge \alpha} \otimes X$) using 2. To see that this has the desired effect, let $\Pi_\lambda$ be the projector onto the eigenspace of $H$ of eigenvalue $\lambda$, and observe that 
\begin{align*}
    \frac{1}{M}\int_0^M \Pi_{\ge \alpha} d\alpha
    &= \frac{1}{M}\int_0^M \sum_{\lambda \ge \alpha}\Pi_{\lambda} d\alpha 
    \\&= \frac{1}{M}\sum_{\lambda \ge 0}\Pi_{\lambda}\int_0^{\lambda}  d\alpha
    \\&= \frac{1}{M}\sum_{\lambda \ge 0}\lambda\Pi_{\lambda}
    =\frac{1}{M} H
\end{align*}

\end{proof}
We are now ready to prove Theorem~\ref{thm:globalLH}

\begin{proof}[Proof (of Theorem~\ref{thm:globalLH})]
Let $H=\sum_i h_i$ be a globally stoquastic Hamiltonian and  let $M=2\sum_i ||h_i||$.
By Lemma~\ref{lem:globaldecomposition}, $H+\beta I$ has a decomposition of the form of Eq.~ \eqref{eq:globaldecomposition}. 
Note that for all $j=0,1 \ldots m'$, $||H_j|| \le M$ which can be argued as follows. $j=0$: Recall from Lemma~\ref{lem:globaldecomposition} that $\beta=-\sum_i \norm{\diag(h_i)}$. Then 
 \[\norm{H_0} =\norm{-\sum_i\diag (h_i)+\beta} \le 2\sum_i \norm{\diag(h_i)} \le M\]
  $j>0$: Since $H_{S,x}$ is diagonal, there exists some $y$ such that $\norm{H_{S,x}} = |\bra{y}H_{S,x}\ket{y}|=|\bra{xy}H\ket{\bar{x}y}|$. Then
 \begin{align}
     \norm{H_{S,x}} &=|\bra{xy}H\ket{\bar{x}y}| \notag \\ 
     &=|\frac{1}{2}(\bra{xy}+\bra{\bar{x}y}) H (\ket{xy}+\ket{\bar{x}y}) \notag \\
     &\qquad-\frac{1}{2}(\bra{xy}H\ket{xy}+\bra{\bar{x}y}H\ket{\bar{x}y})| \notag \\
     &\le 2\norm{H}\le M
 \end{align}

Therefore we can construct a verifier circuit that picks a random $j \in \{0,1,\dots m'\}$, applies $U_j^{\dagger}$ (or $I$ when $j=0$) and then applies the stoquastic verifier circuit of Lemma~\ref{lem:verifiermeasure} to effectively measure $X \otimes \frac{1}{M}H_j$ (or $\frac{1}{M}H_0$ if $j=0$).
On input $\ket{\psi}$, the overall circuit succeeds with probability \begin{align}
\bra{\psi} \frac{1}{m'+1}\sum_{j=1}^{m'}  \frac{I+\frac{1}{M} U_j X \otimes H_j U_j^{\dagger}}{2}\ket{\psi}+\notag \\ +\frac{1}{m'+1}\bra{\psi}\frac{I+\frac{1}{M}H_0}{2}\ket{\psi} \notag \\
=\frac{1}{2}\left(1-\frac{1}{m'M}\bra{\psi}H+\beta I\ket{\psi}\right).
\end{align} An optimal witness $\ket{\psi}$ is a ground state of $H$. Thus we see that there is a $1/{\rm \poly}(n)$ gap between the acceptance probabilities of YES and NO instances of the local Hamiltonian problem as is required for the problem to be in \textbf{StoqMA}.
\end{proof}

\noindent {\em Remarks}: 

One can observe that Theorem \ref{thm:globalLH} can be extended to $O(\log(n))$-local Hamiltonians as long as there are at most $m\le {\rm poly}(n)$ terms $h_i$. \\

Eq.~\eqref{eq:globaldecomposition} in Lemma \ref{lem:globaldecomposition} gives an interesting decomposition of a globally stoquastic Hamiltonian, extending the decomposition for termwise stoquastic Hamiltonians in Lemma 2 in \cite{BBT:stoq}. This decomposition is {\em not} a termwise stoquastic decomposition as in Definition \ref{def:termStoq}, since the terms $U_j (-X\otimes H_j) U_j^{\dagger}$ are not necessarily local due the presence of the classical circuit $U_j$.

We note that the decomposition in Lemma \ref{lem:globaldecomposition} could also be used in the Trotterization of $\exp(-\beta H)$ for a globally-stoquastic Hamiltonians, used in some Monte Carlo algorithms. If $H=\sum_j h_j$ is not term-wise stoquastic, then a Trotterized step like $\exp(-\tau h_j)$ does not necessarily have the property that $\bra{x_{i+1}} \exp(-\tau h_j) \ket{x_{i}} \geq 0$ as $h_j$ is not stoquastic. Using the decomposition in Eq.~\eqref{eq:globaldecomposition} into terms $U_j(-X \otimes H_j) U_j^{\dagger}$, one can absorb the classical transformation $U_j$ into $\ket{x_i}$ and $\ket{x_{i+1}}$ and thus
\begin{align}
    \bra{x_{i+1}} \exp(\tau U_j (X \otimes H_j) U_j^{\dagger}) \ket{x_i}= \notag \\ \bra{x'_{i+1}} \exp(\tau X \otimes H_j) \ket{x'_i} \geq 0.
\end{align}

One should note that for globally-stoquastic Hamiltonians one can also use Eq.~\eqref{eq:def-weight} in Appendix \ref{sec:signProblem} as the nonnegative weight in a Monte Carlo algorithm.

An additional observation is this. For a termwise-stoquastic Hamiltonian $H$ the problem of deciding whether the Hamiltonian is {\em frustration-free}, --i.e. there is a state which is the ground state of each stoquastic term in the termwise decomposition--was shown to be \textbf{MA}-complete \cite{BBT:stoq, BravyiTerhalMA}. This decision problem was called stoquastic $k$-SAT. If, in addition, the promise gap (between YES and NO instances) is constant (rather than $1/{\rm poly}(n)$), this problem, was further shown to be \textbf{NP}-complete in \cite{AG:NP}.  
For globally-stoquastic Hamiltonians one cannot easily formulate such decision problem as we do not have a termwise stoquastic decomposition. However, one could be inspired by the decomposition in Eq.~\eqref{eq:globaldecomposition} and ask about the hardness of deciding whether for a set of nonnegative projectors $U_i \Pi_i U_i^{\dagger}$, where $U_i$ is an efficiently given classical Clifford circuit comprised of CNOT and X gates and $\Pi_i$ is a $O(1)$-local nonnegative projector (with only nonnegative matrix elements), there is a state such that 
\begin{align}
  \forall i,\;  U_i \Pi_i U_i^{\dagger} \ket{\psi}=\ket{\psi}?
  \label{eq:MAproblem?}
\end{align}
We consider it quite likely that this problem is also contained in \textbf{MA} (assuming the same promise as for stoquastic $k$-SAT) following the proof in \cite{BBT:stoq, BravyiTerhalMA}, but we have not examined this in detail \footnote{One could also let $U_i$ be a classical ${\rm poly}(n)$ circuit: one may suspect that this problem is no harder.}. As stoquastic $k$-SAT is \textbf{MA}-complete, it would imply the existence of a reduction of this problem to stoquastic $k$-SAT itself.

Deciding if a globally stoquastic Hamiltonian $H$ with a decomposition as in equation \eqref{eq:globaldecomposition} is frustration-free is equivalent to deciding if there is a state $\ket{\psi}$ satisfying equation \eqref{eq:MAproblem?}, where $U_i$ is as in \eqref{eq:globaldecomposition} and $\Pi_i$ is the projector onto the ground space of $-X \otimes H_i$. The problem in the previous paragraph allows for a more general form of the unitary $U_i$, but on the other hand the projectors $\Pi_i$ will not typically be $O(1)$ local.
The reason for these differences is that in the proof of containment in \textbf{MA} of \cite{BBT:stoq, BravyiTerhalMA}, the verifier performs a random walk on a subset of bitstrings (so called good strings) and accepts as long as she does not reach a bad string. In order to perform a single step in the walk the verifier has to be able to compute 

\begin{equation} \label{step}
P_{x \to y}^i=\sqrt{\frac{\bra{y}U_i\Pi_iU_i^{\dagger}\ket{y}}{\bra{x}U_i \Pi_i U_i^{\dagger}\ket{x}}}\bra{y}(I-\beta H)\ket{x}
\end{equation}
These probabilities are efficiently computable in the case where $U_i$ is a Clifford circuit and $\Pi_i$ is an $O(1)$ local projector.

\section{Determining whether a local Hamiltonian is globally or termwise stoquastic}
\label{sec:fixedbasis}
As we have seen in section \ref{global_vs_termwise} and section \ref{sec:globLH}, the notions of global and termwise stoquasticity can be distinct, although the local Hamiltonian problem is contained in \textbf{StoqMA} for both these classes via Theorem \ref{thm:globalLH}. In this section we shift gears and study the problem of deciding whether a given Hamiltonian is globally or termwise stoquastic. We find that deciding the former can be a significantly harder task than the latter. 

As we mentioned earlier, the stoquastic local Hamiltonian problems we just studied can be seen as a formalisation of the problem of finding the ground energy of some termwise or globally stoquastic Hamiltonian. In any practical scenario however, we would ideally like to be able to efficiently check whether a given Hamiltonian is indeed stoquastic or not before performing some quantum Monte Carlo simulation. In order to formalise this problem we define \glob{} and $m$\term{} as the problems of deciding whether a local Hamiltonian is globally stoquastic and respectively whether it is $m$-termwise stoquastic and ask how hard or easy it is to decide these problems.

First we ask about the efficiency with which one can determine whether a $k$-local Hamiltonian is $m$-termwise stoquastic with $m=O(1)$. In \cite{Marvian2018} it was shown that $m\term \in {\bf P}$, while \cite{KT:stoq} gave a simple explicit strategy for two-local Hamiltonians. We provide the proof for completeness:

\begin{proposition} \label{prop:termInP}
	One can efficiently determine whether a $k$-local Hamiltonian acting on $n$ qubits is $m$-termwise stoquastic where $k \leq m=O(1)$, or $m\term \in \Poly{}$.
	
\end{proposition}

\begin{proof}
	As in the proof of Lemma \ref{lem:bd}, we write $H=\sum_S M^{(S)}$ with $|S| \leq k$. We can ignore the purely diagonal terms in $H$ which flip no qubits, $(S=\emptyset)$. Each $M^{(S)}$ corresponds to different non-zero off-diagonal matrix elements, and thus termwise stoquasticity of $H$ implies termwise stoquasticity of each $M^{(S)}$. 
	For a particular set $S$, we can consider the set of $m$-local stoquastic matrices which only flip the qubits in $S$: this set is spanned by ${\binom{n-|S|}{m-|S|}}2^{m}/2={\rm poly}(n)$ matrices which are products of Pauli operators\footnote{This comes from choosing $m-|S|$ locations outside of $S$ which can be either $Z$ or $I$, and allowing each position in $S$ to be either $X$ or $Y$. The final factor of 2 is because there must be an even number of $Y$ terms for the matrix to be real.} , hence there will be ${\rm poly}(n)$ extremal points, call them ${\bf o}_i(S)$. One thus needs to solve the problem whether there are $p_i \geq0$ such that $M^{(S)}=\sum_i p_i {\bf o}_{i}(S)$ which is a linear feasibility problem. If there exists a (feasible) solution to this program, we move to the next term $M^{(S)}$ etc. until we have found solutions for all terms or for at least one value of $S$ there exists no feasible solution, in which case we output `not $m$-termwise stoquastic'. 
\end{proof}

If we bound the degree of each qubit in the interaction hypergraph, we have shown that global stoquasticity and $O(1)$-stoquasticity are equivalent in Lemma \ref{lem:bd}. In fact in such cases, instead of finding a convex decomposition as in Proposition \ref{prop:termInP}, it is simpler to test for global stoquasticity directly:

\begin{proposition}
	\label{lem:constantinteractions}
	If each qubit of a $k$-local Hamiltonian $H$ interacts with at most $l=O(1)$ other qubits and $k=O(1)$, then the problem of deciding where the Hamiltonian is globally stoquastic can be solved efficiently, that is, \GlobalStoq{} $\in$ \Poly.
\end{proposition}
\begin{proof}
	Let $S \subseteq\{1,\dots, n\}$ be a subset of size $k$. 
	Consider all matrix entries $\bra{x}H\ket{y}$ where $x$ and $y$ differ only in $S$.  
	There are at most $kl=O(1)$ terms in the Hamiltonian which can contribute to these entries, so it is easy to check all these matrix entries are non-positive.
	It suffices to repeat this for all $\binom{n}{k}= \poly(n)$ subsets of size $k$, since $\bra{x}H\ket{y}=0$ for any $x,y$ which differ in more than $k$ places.
\end{proof}

Physically realistic Hamiltonians will typically be geometrically local, where the qubits are distributed $\Omega(1)$ apart in $\mathbb{R}^d$ and only interact with qubits $O(1)$ away. These Hamiltonians have the property that each qubit interacts with at most $O(1)$ other qubits and hence for these classes, global stoquasticity is equivalent to O(1)-termwise stoquasticity by Lemma \ref{lem:bd} and we can efficiently test for both local and global stoquasticity.

\subsection{Hardness of \texorpdfstring{$\glob$}{GlobStoq}}
\label{sec:hardglob}
In this section we will argue that $\glob$ is \textbf{coNP}-complete. To show this, we will consider the well-known \textbf{NP}-complete problem of determining the lowest energy of an Ising Hamiltonian with local fields, and show that the complement of this problem reduces to $\glob$. 

\begin{theorem}[Planar Spin Glass ({\sf PS}) \cite{Barahona_1982}]\label{Barahona}
	Given a planar graph $G=(V,E)$ and an integer $K$. Deciding whether there exists a configuration of $S_i\in\{+1,-1\}$ such that $$\sum_{(i,j)\in E}S_iS_j+\sum_iS_i\leq K$$ is an \textup{\textbf{NP}}-complete problem.
\end{theorem}

\begin{theorem}\label{thrm:Globalstoqu_conp} 
	The problem of deciding whether a $k$-local Hamiltonian is globally stoquastic (\GlobalStoq{}) is \textup{\textbf{coNP}}-complete, that is
		\GlobalStoq{} is in \coNP{} for $k = O(\log n)$ and is \coNP-hard for $k \ge 3$. 
	\label{thm:coNP}
\end{theorem}

\begin{proof}
	It is straightforward to see that \glob{} $\in$ \textbf{coNP}, as there exists an efficiently verifiable witness for its NO instances in the form of bitstrings $x,y\in\{0,1\}^n$ for which  $\bra{x}H\ket{y}>0$. Since $H$ is a $k$-local Hamiltonian, for $k=O(\log  n)$ we can efficiently evaluate such matrix elements and thus verify the witness.
	
	We proceed by showing that \glob{} is  \textbf{coNP}-hard and hence \textbf{coNP}-complete by reducing the complement of {\sf PS} (Theorem \ref{Barahona}) to it. Given an instance of {\sf PS} with graph $G=(V,E)$ such that $|V|=n$, define the $n+1$ qubit Hamiltonian
	\begin{align}\label{conp_proof}
	H=(K+\epsilon^+)X_{n+1}&-\sum_{(i,j)\in E}Z_iZ_jX_{n+1} \nonumber \\&-\sum_{i\in V}Z_iX_{n+1},
	\end{align}
	where $0<\epsilon^+<1$. 
 	Since the first $n$ qubits are only acted on by $Z_i$ terms and the last qubit is only acted on by a $X_{n+1}$ term, $\bra{x}H\ket{y}$ is non-zero only when $x$ and $y$ match on the first $n$ bits and differ on the last.
	
	For $x,y$ pairs of this form, let $x'$ denote the first $n$ bits of $x$ and $y$, so that
	\begin{align*}
	\bra{x}H\ket{y}&=  \bra{x'}(K+\epsilon^+)-\sum_{(i,j)\in E}Z_iZ_j \\
	& \quad-\sum_{i\in V}Z_i\ket{x'}\bra{x_{n+1}}X_{n+1}\ket{y_{n+1}} . \nonumber
	\end{align*}
	Note that $\bra{x_{n+1}}X_{n+1}\ket{y_{n+1}}=1$, and convert from $\{0,1\}$ to $\{-1,+1\}$ by setting $S_i=2x_i'-1=\bra{x'}Z_i\ket{x'}$, to get
    \[\bra{x}H\ket{y}=(K+\epsilon^+)-\sum_{(i,j)\in E}S_iS_j-\sum_{i\in V}S_i.\]
	
	Therefore $H$ is globally stoquastic if and only if for all $\{S_i\}$
	\[\sum_{(i,j)\in E}S_iS_j+\sum_{i\in V}S_i\geq K+\epsilon^+\]
	
	For a YES instance of {\sf PS}, there exists $\{S_i\}$ such that $\ \sum_{(i,j)\in E}S_iS_j+\sum_{i\in V}S_i\leq K<K+\epsilon^+$, implying a NO instance of \GlobalStoq.
	
	For a NO instance of {\sf PS}, for all  $\{S_i\}$,  $\sum_{(i,j)\in E}S_iS_j+\sum_{i\in V}S_i\geq K+1> K+\epsilon^+$ implying a YES instance of \GlobalStoq.

\end{proof}

{\em Remark}: Although deciding \glob{} is \textbf{coNP}-complete for $H$ as in Eq.~(\ref{conp_proof}) and thus is \textbf{coNP}-complete in general, such Hamiltonians can be sign-cured by means of a single-qubit Hadamard transformation on the qubit with label $n+1$. However, in Section \ref{gadgets} we will introduce so-called gadgets i.e additional terms which can be added to the Hamiltonian in order to restrict the set of possible sign-curing transformations. Under the addition of such terms the problem of sign-curing the type of Hamiltonians we just considered also becomes a hard computational problem.

\section{Global stoquasticity via local basis changes}
\label{sec:stoq-lb}

In the previous section we considered the difficulty of determining if a Hamiltonian is globally stoquastic in a fixed basis. In practice, one could try to find a suitable basis in which to represent the Hamiltonian such that it does not suffer from the sign problem if it did in the computational basis. Previous work \cite{ Marvian2018,KT:stoq,Klassen2019} has shown that it is $\NP$-hard to decide if there is a local change of basis such that a Hamiltonian is termwise stoquastic. Now we combine both questions, and consider the problem of determining if there exists a local change of basis such that a Hamiltonian is globally stoquastic. We start with a definition:
\begin{definition}
	Let $\mathcal{U}$ be a family of unitaries. Then $\mathcal{U}$-{\sf GlobalStoq} is the following problem:
	
	\textbf{Input:} Local Hamiltonian $H$
	
	\textbf{Problem:} Decide if there is a unitary  $U \in \mathcal{U}$ such that $UHU^{\dagger}$ is globally stoquastic.
\end{definition}
Observe that {\sf GlobalStoq} is the special case of $\mathcal{U}$-{\sf GlobalStoq} where $\mathcal{U}$ only contains the identity matrix. Our result is the following.

\begin{theorem}
	\label{thm:Sigmatwo}
	Let $\mathcal{U}$ be the set of unitaries which are products of single qubit unitaries. 
	Then $\mathcal{U}$-\GlobalStoq{} is $\Sigmatwo$-hard for 3-local Hamiltonians.
	$\mathcal{U}$-\GlobalStoq{} is contained in $\Sigmatwo$ if there is the additional promise that in a YES instance the curing unitary can be efficiently described in such a way that allows one to efficiently compute matrix entries of $UHU^{\dagger}$. 
\end{theorem} 

To prove this theorem, we will need three different ingredients:
\begin{itemize}
	\item a convenient $\Sigmatwo$-complete problem to reduce from
	\item a "gadget" Hamiltonian construction to restrict the possible basis change unitaries
	\item a Hamiltonian which encodes our $\Sigmatwo$-complete problem.
\end{itemize}

\subsection{Complete problem for \texorpdfstring{$\Sigmatwo$}{Sigma2}}
The complexity class $\mathbf{\Sigmatwo}=\mathbf{NP^{NP}}$ sits in the second level of the polynomial hierarchy, and is the class of problems that can be solved by a polynomial time non-deterministic Turing machine (an $\NP$ machine), which can make queries to an oracle to $\NP$.

Equivalently, $\Sigmatwo$ is the class of problems for which there exists a polynomial time verifier $V$ such that a problem instance is a YES instance iff 
\[\exists x, \: \forall y \: V(x,y) \text{ accepts} \]
where $x$ and $y$ are bit strings of polynomial length. 

Recall that \GlobalStoq{} is a YES instance if there exists a basis change, such that all matrix entries of $H$ are non-positive. If we let $x$ be a description of a basis change and $y$ the specification of a matrix entry, then we see that  $\GlobalStoq{} \in \Sigmatwo$ if the basis change can be described by polynomially sized bit string $x$ in such a way that allows the verifier to compute a matrix entry in polynomial time.
Since the basis change is a product of 1-local unitaries, one might expect such an efficient description to exist, but it might be the case that these unitaries need to be specified to very high precision - which is why this condition is included in the statement of Theorem~\ref{thm:Sigmatwo}.

The complement of $\Sigmatwo$ is $\Pitwo$. A canonical complete problem for $\bf \Pi_2^p$ is $\forallexistSAT$ \cite{Schaefer2002}:
\begin{definition}[$\forallexistSAT$] 
	
	\textbf{Input:} Boolean formula $\mathcal{C}(x,y)$ in conjunctive normal form with 3 literals per clause.
	
	\textbf{Problem:} Decide if for all $x$, there exists a $y$ such that $\mathcal{C}(x,y)$ is true.
\end{definition}
The complement of $\forallexistSAT$ is complete for $\Sigmatwo$ (since $\Sigmatwo = \class{co-}\Pitwo$).
We could reduce directly from this problem to $\mathcal{U}$-\GlobalStoq, using the same construction as in Section~\ref{sec:Sigmatwohardness}, to obtain $\Sigmatwo$-hardness, but the resulting Hamiltonian would be 4-local. In order to achieve a hardness result for a 3-local Hamiltonian, we first find a more convenient complete problem for $\Sigmatwo$.
\begin{definition} [\minmaxSAT]
	\textbf{Input:} Boolean formula $\mathcal{C}(x,y)$ in conjunctive normal form with at most 2 literals per clause; and an integer $k$.
	
	\textbf{Problem:} Decide if for all $x$, there exists a $y$, such that  $\mathcal{C}(x,y)$ satisfies at least $k$ clauses.
\end{definition}
The problem \minmaxSAT{} is related to $\forallexistSAT$, in the same way that \problem{MAX-2-SAT} is related to \problem{3-SAT}.
In fact we can prove \minmaxSAT{} is $\Pitwo$-complete, by reducing from $\forall\exists$-SAT, with exactly the same proof method as the reduction from \problem{3-SAT} to \problem{MAX-2-SAT}.

\begin{lemma}
	\label{lem:minmaxsat}
	\minmaxSAT{} is $\Pitwo$-complete.
\end{lemma}
\begin{proof}
    For containment in $\Pitwo$, we construct a verifier circuit $V(x,y)$ which evaluates all of the clauses of $C$ and accepts if at least $k$ clauses are satisfied. Then \minmaxSAT{} is a YES instance iff there exists an $x$ such that for all $y$, $V(x,y)$ accepts and hence \minmaxSAT{}$\in \Pitwo$.

	To prove hardness, let $\mathcal{C}(x,y)$ be a formula specifying an instance of $\forallexistSAT$ consisting of $m$ clauses on $n+l$ bits, where $x \in \{0,1\}^n$ and $y \in \{0,1\}^l$. We construct a formula $\mathcal{C}'$ of $10m$ clauses on $n+l+m$ bits, with at most 2 literals per clause, such that if $x\in \{0,1\}^n$ and $y \in \{0,1\}^l$ satisfies $\mathcal{C}$, then there exists $z \in \{0,1\}^m$ such that $(x,y,z)$ satisfies $7m$ clauses of $\mathcal{C}'$; and if $(x,y)$ does not satisfy $\mathcal{C}$ then $(x,y,z)$ satisfies strictly less than $7m$ clauses, for any $z \in \{0,1\}^m$.
	
	For each clause containing three literals $(a_i \vee b_i \vee c_i)$, replace it with 10 clauses $(a_i), (b_i), (c_i), (d_i), (\neg a_i \vee \neg b_i), (\neg a_i \vee \neg c_i), (\neg b_i \vee \neg c_i),(a_i \vee \neg d_i), (b_i \vee \neg d_i), ( c_i \vee \neg d_i)$ 
	
	Then if the original clause $(a_i \vee b_i \vee c_i)$ is satisfied, there is a choice of $d_i$ such that 7 of these clauses are satisfied, but there is no choice of $d_i$ such that more than 7 are satisfied.
	And if the clause is not satisfied (all 3 of $a_i,b_i,c_i$ are false), then any choice of $d_i$ will satisfy at most six of these clauses.
	
	Then $\forallexistSAT$ for $\mathcal{C}$ reduces to \minmaxSAT{} for $\mathcal{C}'$ and $k=7m$. Since if for all $x$, there exists $y$ such that $\mathcal{C}(x,y)$ is true, then for all $x$ there exists $y,z$ such that $k=7m$ clauses of $\mathcal{C}'$ are satisfied.
	And if there exists $x$ such that for all $y$ $\mathcal{C}(x,y)$ is unsatisfied, then there exists $x$ such that for all $y,z$, strictly less than $7m$ clauses of $\mathcal{C'}$ are satisfied.
\end{proof}

An immediate consequence of Lemma~\ref{lem:minmaxsat} is that the complement of \minmaxSAT, which we call $\neg$\minmaxSAT, is $\Sigmatwo$-complete.
\begin{definition} [$\neg$\minmaxSAT]
	\textbf{Input:} Boolean formula $\mathcal{C}(x,y)$ in conjunctive normal form with at most 2 literals per clause; and an integer $k$.
	
	\textbf{Problem:} Decide if there exists $x$, such that for all $y$,  $\mathcal{C}(x,y)$ violates at least $k$ clauses.
\end{definition}
Note that an equivalent definition of \minmaxSAT{} is to decide if $C(x,y)$ satisfies at least $k$ clauses.

\subsection{Gadgets} \label{gadgets}

In \cite{Klassen2019} two-local Hamiltonian gadget terms were constructed such that, when these terms are added to a Hamiltonian $H$, a local change of basis can make the total Hamiltonian stoquastic if and only if the original Hamiltonian $H$ can be made stoquastic by a restricted type of basis of change.

In the following, we use $W$ to denote the single-qubit Hadamard matrix and for $ x \in \{0,1\}^n$ we define  
\[W(x)= \bigotimes_{i=1}^n W_i^{x_i}.\]

\begin{lemma}[Lemma 5.1 of \cite{Klassen2019}]
    \label{lem:oldgadgets}
    Let H be a Hamiltonian on $n$ qubits. For each qubit $u \in \{1,...,n\}$, add three ancilla qubits $a_u, b_u, c_u$ and define the two-local gadget Hamiltonian 
   \begin{align*}
	    G_1=\sum_{u=1}^{n}-&(X_{c_u}+Z_{c_u})-(X_uX_{a_u}+Y_uY_{a_u}+Z_uZ_{a_u})\\&-(3X_{a_u}X_{b_u}+Y_{a_u}Y_{b_u}+2Z_{a_u}Z_{b_u})\\&-(X_{b_u}X_{c_u}+Y_{b_u}Y_{c_u}+Z_{b_u}Z_{c_u})
	\end{align*}
	
	Then the following are equivalent:
	\begin{enumerate}
	    \item there exists a product of single qubit unitaries $U$ such that $U(H\otimes I +G_1)U^{\dagger}$ is a globally stoquastic Hamiltonian.
	    
	    \item there exists $x \in \{0,1\}^n$ such that $W(x)^{\dagger} H W(x)$ is globally stoquastic, where $W(x)=\left(\bigotimes_{i=1}^n W_i^{x_i}\right)$. 
	\end{enumerate}
\end{lemma}

Lemma~\ref{lem:oldgadgets} restricts any sign-curing product unitary to act either as $W$ or $I$ on each qubit.
Now we present a small extension to these gadgets, in the form of an additional gadget Hamiltonian $G_2$, to further restrict the type of possible basis changes on certain qubits. 
For a Hamiltonian acting on $n+l$ qubits, these gadget terms force any product sign curing unitary to act as either $W$ or $I$ on $n$ of the qubits, and only as $I$ on the other $l$ qubits.

\begin{lemma}[Extension of Lemma~\ref{lem:oldgadgets}]
	\label{lem:gadgets}
	Let $H$ be a Hamiltonian on $n+l$ qubits. First add an ancilla qubit $d_j$ for each $j\in \{1,...,l\}$  and define the 2-local gadget Hamiltonian 
	
	\begin{equation}
	    G_2=\sum_{j=1}^{l} -X_{n+j}X_{d_j}+Z_{n+j}Z_{d_j}
	\end{equation}
	
	 Then add three more ancilla qubits $a_u,b_u,c_u$ for each of the $n+2l$ qubits.
	Defining the total Hamiltonian $H_{\rm Had}$ on $4n+8l$ qubits as 
	\[H_{\rm Had} = (H\otimes I +G_2) \otimes I + G_1\]
	where $G_1$ is as defined in Lemma \ref{lem:oldgadgets}, the following are equivalent:
	\begin{enumerate}
		\item there exists a product of single qubit unitaries $U$ such that $U H_{\rm Had} U^{\dagger}$ is globally stoquastic.
		\item there exists $x \in \{0,1\}^{n}$ such that $W(x0^l)^{\dagger} H W(x0^l)$ is globally stoquastic, where $W(x0^l)=\left(\bigotimes_{i=1}^n W_i^{x_i}\right)\otimes \left(\bigotimes_{i=n+1}^{n+l} I_i\right)$. 
	\end{enumerate}  
\end{lemma}
\begin{proof}

	We first prove that 2 $\Rightarrow$ 1. Let $H'= H \otimes I +G_2$. 
	If there exists $x\in \{0,1\}^{n}$ such that $W(x0^l)HW(x0^l)^{\dagger}$ is globally stoquastic, then \[W(x0^{2l})H'W(x0^{2l})^{\dagger}=\left(W(x0^l)HW(x0^l)^{\dagger}\right)\otimes I +G_2\]
	is also globally stoquastic, since the off-diagonal matrix entries of $G_2$ are all non-positive.
	Therefore by Lemma~\ref{lem:oldgadgets}, there exists a unitary $U$, which is a product of single qubit unitaries, such that $UH_{\text{Had}}U^{\dagger}$ is globally stoquastic.
	
	For the converse direction, Lemma~\ref{lem:oldgadgets} implies that if 1. holds then there exists $x\in \{0,1\}^n$ and $y,z \in \{0,1\}^l$ such that $W(xyz) H' W(xyz)^{\dagger}$ is globally stoquastic, where 
 \[W(xyz)= \left(\bigotimes_{i= 1}^{n} W_{i}^{x_i} \right)\otimes\left(\bigotimes_{j=1}^{l} W_{n+j}^{y_{j}} \otimes W_{d_j}^{z_{j}} \right)\]
Suppose for a contradiction that $y_j$ or $z_j$ is 1 for some $j$. Then there is a positive off-diagonal matrix entry $\bra{u}W(xyz) G_2 W(xyz)^{\dagger} \ket{v} >0$ where $u,v \in \{0,1\}^{n+2l}$ differ at location $d_j$. This cannot be cancelled out by $W(xyz)(H\otimes I)W(xyz)^{\dagger}$ (since this acts trivially on qubit $d_j$), implying that $W(xyz) H' W(xyz)^{\dagger}$ is not globally stoquastic. But this is a contradiction and we can conclude that $y=z =0^l$.
 
 Finally if $W(x0^{2l})H'W(x0^{2l})^{\dagger}=W(x0^l)HW(x0^l)^{\dagger}\otimes I +G_2$ is globally stoquastic then $W(x0^l)HW(x0^l)^{\dagger}$ is globally stoquastic. This is because the off-diagonal matrix entries of $G_2$, $\bra{x}G_2\ket{y}$, are non-zero only when $x_{d_j}\neq y_{d_j}$ for some $j$, and so these matrix entries cannot cancel out any of the off-diagonal matrix entries of $W(x0^l)HW(x0^l)^{\dagger}\otimes I$ which acts trivially on qubits $d_j$.
	

\end{proof}
\subsection{Hardness construction}
\label{sec:Sigmatwohardness}

For a 2-SAT formula $\mathcal{C}(a,b)$ in conjunctive normal form, with $a \in \{0,1\}^n$  and $b \in \{0,1\}^l$, 
\[\mathcal{C}=\bigwedge_{k=1}^m \left(c_{k,1} \vee c_{k,2}\right),\]
 we define a corresponding Hamiltonian $H_{\mathcal{C}}$ on $n+l$ qubits:	
 \[ H_{\mathcal{C}}= \sum_{k=1}^m P(c_{k,1}) \otimes P(c_{k,2}) \]
where  
\[P(c)= \begin{cases}
(X_i+I) & \text{ if } c = a_i \\
(Z_i +I) & \text{ if } c = \neg a_i \\
\ket{0}\bra{0}_{n+i} & \text{ if } c = b_i \\
\ket{1}\bra{1}_{n+i} & \text{ if } c = \neg b_i \\
\end{cases}\]
For example, the formula
\[\mathcal{C}_{\text{example}}= (a_1 \vee b_2) \bigwedge (\neg a_2 \vee b_1)\]
corresponds to the Hamiltonian
\[H_{\mathcal{C}_{\text{example}}}=X_1 \otimes\ket{0}\bra{0}_4 +Z_2\otimes \ket{0}\bra{0}_3. \]

We now consider all possible basis changes by Hadamards on the first $n$ qubits and introduce the shorthand  $H_{\mathcal{C}}(x):=W(x0^l)H_\mathcal{C}W(x0^l)^{\dagger}$. For all $x \in \{0,1\}^n$, $H_{\mathcal{C}}(x)$ has the following properties:
\begin{enumerate}
	\item All off-diagonal elements of $H_{\mathcal{C}}(x)$ are non-negative.
	\item The diagonal matrix entries satisfy:
	\[
	\bra{z}\bra{y} H_{\mathcal{C}}(x)\ket{z}\ket{y} \ge \bra{1^n} \bra{y}H_{\mathcal{C}}(x)\ket{1^n}\ket{y}  
	\]
	where $x,z \in \{0,1\}^n$ and $y \in \{0,1\}^l$.
	\item $\bra{1^n} \bra{y} H_{\mathcal{C}}(x)\ket{1^n}\ket{y}$
	is the number of clauses of $\mathcal{C}$ violated by $(x,y)$.
\end{enumerate}
Properties 1. and 2. can be easily checked. 
To see point 3., observe that 
\begin{align*}
\bra{1^n}\bra{y}W(x0^l)P(c) & W(x0^l)^{\dagger}\ket{1^n}\ket{y}\\&=\begin{cases}
0 & \text{ if }c(x,y) \text{ is true} \\
1 & \text{ if }c(x,y) \text{ is false}
\end{cases}
\end{align*}
and so 
\begin{align*}\bra{1^n}\bra{y}W(x0^l)P(c_1)P(&c_2) W(x0^l)^{\dagger}\ket{1^n}\ket{y}\\
&=\begin{cases}
0 & \text{ if }(c_1\vee c_2) \text{ is true} \\
1 & \text{ if }(c_1 \vee c_2) \text{ is false}
\end{cases}
\end{align*}
Therefore, $\bra{1^n}\bra{y}W(x0^l)H_{\mathcal{C}} W(x0^l)^{\dagger}\ket{1^n}\ket{y}$ is equal to the number of clauses in $\mathcal{C}$ that are not satisfied by $(x,y)$, as claimed.

We are now ready to prove Theorem~\ref{thm:Sigmatwo}.
\begin{proof}[Proof. (of Theorem~\ref{thm:Sigmatwo})]
We first prove containment in $\Sigmatwo{}$ under the additional promise of the theorem. Let $y\in\{0,1\}^{2n}$ be the concatenation of two bit strings $y_1,y_2\in \{0,1\}^n$ and let $x$ be a description of a unitary $U_x$ that allows for efficient computation of the matrix entries of $U_xHU_x^{\dagger}$. Then there exists a polynomial size verifier circuit $V$ that accepts on input $x,y$ if $\bra{y_1}U_xHU_x^{\dagger}\ket{y_2} \le 0$ or $y_1=y_2 $. Then $\mathcal{U}$-\GlobalStoq{} is a YES instance iff $\exists x, \: \forall y \: V(x,y)$ accepts, and hence is in $\Sigmatwo$.
	
	

	To prove hardness we consider Hamiltonians of the form 	
	\begin{equation}H=(X \otimes \left(kI-H_{\mathcal{C}}\right)\otimes I+G_2)\otimes I+G_1
	\label{eq:thm4proof}\end{equation}
	where $H_\mathcal{C}$ is as defined above, and $G_1$, $G_2$ are the gadget Hamiltonians of Lemma~\ref{lem:gadgets}. These are chosen so that there exists a product unitary $U$ such that $UHU^{\dagger}$ is globally stoquastic if and only if there exists $x \in \{0,1\}^n$ such that 
	\begin{align*}W(0x & 0^l)\left(X \otimes (kI-H_{\mathcal{C}}) \right) W(0x0^l)^{\dagger} \\
	& = X \otimes \left[kI -W(x0^l) H_{\mathcal{C}}W(x0^l)^{\dagger}\right] \end{align*}
	is globally stoquastic.
	
	This happens if and only if all the off-diagonal terms of  $H_{\mathcal{C}}(x)=W(x0^l)H_{\mathcal{C}}W(x0^l)^{\dagger}$ are non-negative and all the diagonal matrix entries are greater than or equal to $k$.
	
	As observed above, all the off-diagonal elements of $H_{\mathcal{C}}(x)$ are non-negative and the smallest diagonal matrix entry is of the form $\bra{1^n} \bra{y} H_{\mathcal{C}}(x)\ket{1^n}\ket{y}$ for some $y$. Furthermore $\bra{1^n} \bra{y} H_{\mathcal{C}}(x)\ket{1^n}\ket{y}$ is equal to the number of unsatisfied clauses of $\mathcal{C}$.
	Therefore $\mathcal{U}$-\GlobalStoq{} is a YES instance if and only if there exists $x$ such that for all $y$ at least $k$ clauses of $\mathcal{C}$ are unsatisfied.
	That is, $H$ is a YES instance of $\mathcal{U}$-\GlobalStoq{} if and only if $\mathcal{C}$ is a YES instance of $\neg$\minmaxSAT. $\neg$\minmaxSAT is $\Sigmatwo$-complete and so, for the class of Hamiltonians of the form \eqref{eq:thm4proof}, $\mathcal{U}$-\GlobalStoq{} is $\Sigmatwo$-complete.

\end{proof}

\section{Sign-curing by Clifford transformations} \label{sec:Clifford}

In the previous section we considered the problem of sign-curing a Hamiltonian by product unitaries of the form $\bigotimes_{i=1}^n U_i$ where $U_i \in  {\sf U(2)}$. However, these are not the only computationally efficient transformations we could consider. In this section we go beyond these single-qubit transformations and ask which $k$-local Hamiltonians can be mapped onto stoquastic Hamiltonians by means of Clifford transformations. It is apriori not clear that the class of sign-curable Hamiltonians can be enlarged using these transformations in contrast to single qubit unitaries. Here we show that this is indeed the case and give an explicit example of a class of 1D Hamiltonians which can be sign-cured with Clifford unitaries but cannot be sign-cured using product unitaries of the form $\bigotimes_{i=1}^n U_i$ where $U_i \in  {\sf U(2)}$.

We start with a reminder about the definition of the Clifford group.

\begin{definition}[Clifford group] 
	\begin{equation*}
	{\sf C}_n=\{U\in {\sf U(2^n)}\ |\ UPU^{\dagger}\in {\sf P}_n \ \forall P\in {\sf P}_n\}
	\end{equation*}
	where ${\sf P}_n$ is the Pauli group on n qubits, that is, ${\sf P}_n =\{I,X,Y,Z\}^{\otimes n}\times \{\pm 1, \pm i\}$
	\label{def:cliff}
\end{definition}

A simple example of a Hamiltonian which can be sign-cured with Clifford transformations is a Hamiltonian which is a sum of commuting Pauli terms $H=\sum_i J_i P_i$ with Paulis $\forall i,j, \;[P_i,P_j]=0$ such as a stabilizer Hamiltonians. $H$ is then obviously computationally-stoquastic as it can be transformed by the adjoint action of an element of the Clifford group (i.e. a Clifford transformation $C$) with $C^{\dagger} P_i C=\tilde{P}_i(Z)$ where $\tilde{P}_i(Z)$ is a tensor product of $Z$ gates, so the transformed Hamiltonian $\tilde{H}=\sum_i J_i \tilde{P}_i(Z)$ is diagonal and thus stoquastic. The existence of a Clifford transformation which performs this mapping is attributed to the fact that Clifford transformations conserve the product relations among the Paulis. Thus, modulo $\pm i$ and $\pm 1$ signs, we can conclude that there is some Clifford circuits performing the transformations as long as the product relations among the Paulis $\tilde{P}_i(Z)$ are those of $P_i$. In the above example one can check explicitly that these product relations are indeed conserved. 

More generally, for a Hamiltonian $H=\sum_i J_i P_i$, one specifies a set $S_{\rm indep}$ of independent Paulis among the $P_i$ such that other Pauli terms in $H$ are obtained by taking products of these independent elements (modulo $\pm 1, \pm i$ prefactors). If we have a Hermiticity-preserving injective map $g$ such that each
$P \in {\cal S}_{\rm indep}$ is mapped to a $P'$, i.e. $g(P)=P'$ so that $g(P_1 P_2)=g(P_1) g(P_2)$ for $P_1, P_2 \in {\cal S}_{\rm indep}$, then there is always a realization for $g$ which is a Clifford transformation. Hence specifying the action of such map acting on the independent set $S_{\rm indep}$ means that a Clifford transformation will exist (and can be found explicitly as a symplectic transformation \cite{KS:sym}).
\usetikzlibrary{arrows}
 \begin{figure*}
 	\centering
 \usetikzlibrary{arrows}
	\begin{tikzpicture}[scale=2, every node/.style={scale=2}]

 \node[circle, font=\large, fill=red,inner sep=0pt, minimum size=3pt,label={left:\Large$X_{i-1}X_i$}] (d)   at (-7,0) {}; 
 \node[circle, font=\LARGE, fill=red,inner sep=0pt, minimum size=3pt, label={left:\Large$Y_{i-1}Y_i$}]  (e)  at (-7,-3) {}; 
 \node[circle, fill=red,inner sep=0pt, minimum size=3pt, label={left:\Large$Z_{i-1}Z_i$}]  (f)  at (-7,3) {};
 \node[circle, fill=red,inner sep=0pt, minimum size=3pt, label={[label distance=4mm] 90:\Large$X_{i}X_{i+1}$}]  (g)  at (0,0) {}; 
 \node[circle, fill=red,inner sep=0pt, minimum size=3pt, label={[label distance=4mm] 90 :\Large$Y_{i}Y_{i+1}$}]  (h)  at (0,-3) {}; 
 \node[circle, fill=red,inner sep=0pt, minimum size=3pt, label={above:\Large$Z_{i}Z_{i+1}$}] (i)  at (0,3) {};
 
 \node[circle, fill=red,inner sep=0pt, minimum size=3pt, label={right:\Large$X_{i+1}X_{i+2}$}]  (j)  at (7,0) {}; 
 \node[circle, fill=red,inner sep=0pt, minimum size=3pt, label={right:\Large$Y_{i+1}Y_{i+2}$}] (k)  at (7,-3) {}; 
 \node[circle, fill=red,inner sep=0pt, minimum size=3pt, label={right:\Large$Z_{i+1}Z_{i+2}$}] (l)  at (7,3) {};

 \draw[black] (d) -- (h) ;
 \draw[black] (d) -- (i) ;
 \draw[black] (f) -- (h) ;
 \draw[black] (f) -- (g) ;
 \draw[black] (e) -- (g) ;
 \draw[black] (e) -- (i) ;
 
 \draw[black] (g) -- (l) ;
 \draw[black] (g) -- (k) ;
 \draw[black] (h) -- (j) ;
 \draw[black] (h) -- (l) ;
 \draw[black] (i) -- (j) ;
 \draw[black] (i) -- (k) ;
 \draw[implies-implies,double equal sign distance] (0,-3.5) -- (0,-5);

	\node[circle, fill=red,inner sep=0pt, minimum size=3pt,label={left:\Large$\Delta_{XX} Z_{i-1}$}]  (d)  at (-7,-9) {}; 
	
	\node[circle, fill=red,inner sep=0pt, minimum size=3pt, label={left:\Large$\Delta_{YY} Z_i$}]  (e)  at (-7,-12) {}; 
	
	\node [circle, fill=red,inner sep=0pt, minimum size=3pt, label= {[align=left] 180: \Large$-\Delta_{YY}\Delta_{XX}$ \\ $Z_{i-1}Z_i$}]  (f)  at (-7,-6) {};
	
	
	\node[circle, fill=red,inner sep=0pt, minimum size=3pt, label={[label distance=4mm] 90:\Large$\Delta_{XX}X_{i}X_{i+2}$}]  (g)  at (0,-9) {}; 
	
	\node[circle, fill=red,inner sep=0pt, minimum size=3pt, label={[label distance=1mm]  270 :\Large$\Delta_{YY}X_{i-1}X_{i+1}$}]  (h)  at (0,-12) {};
	
	\node[circle, fill=red,inner sep=0pt, minimum size=3pt, label={above:\Large$-\Delta_{YY}\Delta_{XX}X_{i-1}X_{i}X_{i+1}X_{i+2}$}]  (i)  at (0,-6) {};

	\node[circle, fill=red,inner sep=0pt, minimum size=3pt, label={right:\Large$\Delta_{XX}Z_{i+1}$}]  (j)  at (7,-9) {}; 
	
	\node[circle, fill=red,inner sep=0pt, minimum size=3pt, label={right:\Large$\Delta_{YY}Z_{i+2}$}]  (k)  at (7,-12) {}; 
	
	
	\node[circle, fill=red,inner sep=0pt, minimum size=3pt, label={[align=left] 0:\Large$-\Delta_{YY}\Delta_{XX}$ \\ \Large$Z_{i+1}Z_{i+2}$}]  (l)  at (7,-6) {};
	
	\draw[black] (d) -- (h) ;
	\draw[black] (d) -- (i) ;
	\draw[black] (f) -- (h) ;
	\draw[black] (f) -- (g) ;
	\draw[black] (e) -- (g) ;
	\draw[black] (e) -- (i) ;
	
	\draw[black] (g) -- (l) ;
	\draw[black] (g) -- (k) ;
	\draw[black] (h) -- (j) ;
	\draw[black] (h) -- (l) ;
	\draw[black] (i) -- (j) ;
	\draw[black] (i) -- (k) ;		
	\end{tikzpicture} 
		\caption{Top: Graph summarising the (anti-) commutation relations between the Pauli terms in the XYZ Heisenberg model. An edge between two vertices representing Paulis indicates anticommutation between them while the absence of an edge indicates that they commute. Bottom: Representation of these terms when transformed under a Clifford transformation which preserves the commutation structure. $\Delta_{XX}$, $\Delta_{YY}$ $=\pm1$ are chosen for each edge $(i,i+1)$ of $H$ independently.}
	\label{fig:comm}
\end{figure*}
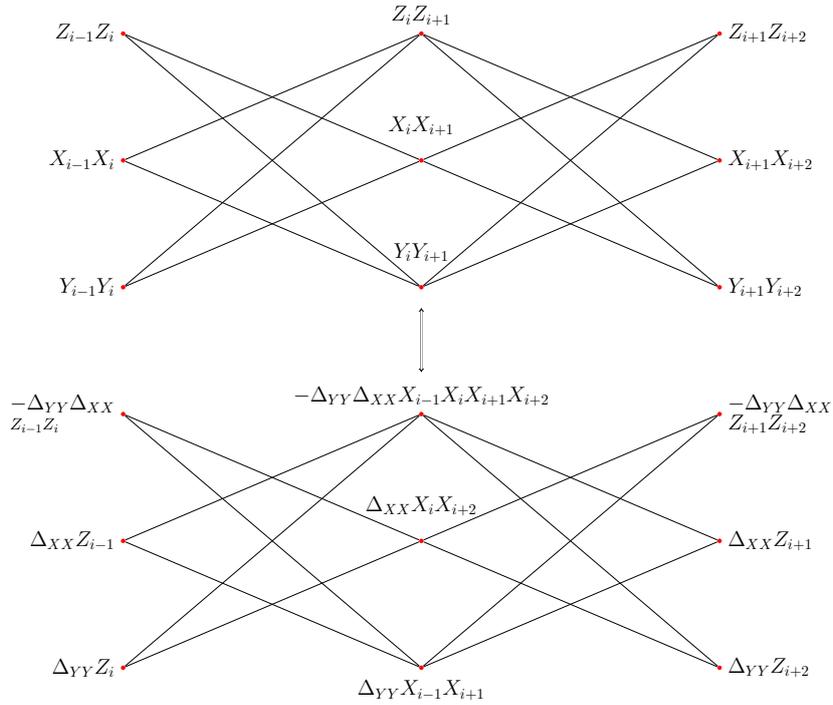
Assume that $\tilde{H}=C^{\dagger} H C$ with Clifford $C$ and let $\tilde{H}$ be globally stoquastic. We observe that if we wish to estimate thermal or ground state properties of $H$ via a path integral quantum Monte Carlo method, see Appendix \ref{sec:signProblem}, then we can simply use $\tilde{H}$ in this simulation: the Pauli terms in Hamiltonian $\tilde{H}$ may have high-weight, but this does not make the Monte Carlo method inefficient (although it may require different Metropolis update rules), since the number of such Pauli terms is the same as in $H$ and thus matrix elements of $\tilde{H}$ can still be determined efficiently.

We now will provide an example of a class of 1D disordered Heisenberg Hamiltonians for which one can prove that sign-curing using single-qubit unitary transformations is not possible, while Clifford transformations do sign-cure the Hamiltonians in this class.

Consider the disordered Heisenberg XYZ model on an open boundary 1D chain of $n$ qubits with arbitrary coupling coefficients:
\begin{align}
H=\sum_{i=1}^{n-1}\alpha_{i \ i+1}^{XX}X_iX_{i+1} &+\alpha_{i\ i+1}^{YY}Y_iY_{i+1} \nonumber\\ &+\alpha_{i\ i+1}^{ZZ}Z_iZ_{i+1}.
\label{eq:heis}
\end{align}

In what follows, we will refer to pairs $(i,i+1)$ as edges of $H$. It is known that a 1D translationally-invariant Heisenberg XYZ model has no sign problem \cite{Suzuki} and we review the argument in Appendix \ref{app:XYZ}. It was shown in \cite{KT:stoq} that if we seek to sign-cure an arbitrary XYZ Heisenberg model by single-qubit (local) basis changes, then single-qubit Clifford transformations suffice. However, not all 1D XYZ models can be sign-cured by single-qubit unitary transformations and we give some explicit examples where such transformations fall short in Appendix \ref{sec:fs}. Given these examples, the following theorem then shows that Clifford transformation are more powerful than single-qubit unitary transformations.

\begin{theorem}\label{clifford2}
	The (open boundary) 1D disordered XYZ Heisenberg model in Eq.~(\ref{eq:heis}) is $4$-termwise stoquastic by a Clifford transformation if, either for all even edges we have
	\begin{equation}
	\forall i,\;\alpha_{2i \ 2i+1}^{XX} \times \alpha_{2i \ 2i+1}^{YY}\times \alpha_{2i \ 2i+1}^{ZZ} \geq 0,  
	\label{eq:cond1} 
	\end{equation}
	or, for all odd edges we have
	\begin{equation}
	\forall i,\; \alpha_{2i-1 \ 2i}^{XX} \times \alpha_{2i-1 \ 2i}^{YY}\times \alpha_{2i-1 \ 2i}^{ZZ} \geq 0.   
	\label{eq:cond2}
	\end{equation}
\end{theorem}

\begin{proof}
	We can see the commutation structure of the XX, YY and ZZ terms in the Heisenberg model in Fig.~\ref{fig:comm}. The XX and YY terms can be chosen as the independent set $S_{\rm indep}$.
	If Eq.~(\ref{eq:cond1}) holds, we map the commuting XX, YY, ZZ terms between qubits $2i-1$ and $2i$ (odd edges of $H$) to terms which are tensor products of Z gates, which all mutually commute as required. Since $X_iX_{i+1} Y_iY_{i+1} =- Z_i Z_{i+1}$, we just need to make sure that we choose the 3 terms with the right sign, so as to preserve all product relations. However, terms which are tensor products of Z gates are diagonal, their signs do not matter for stoquasticity so any consistent choice is valid.
	If instead Eq.~(\ref{eq:cond2}) holds, we would have done the same for the terms between qubits $2i$ and $2i+1$ (even edges of $H$). 
	
	Now consider wlog that Eq.~(\ref{eq:cond1}) holds and we need to map the commuting  XX, YY, ZZ terms between qubits on even edges of $H$ to Pauli terms respecting all product relations. One can easily check that we can map these terms to tensor products of X gates, acting on at most 4 qubits, as shown in Fig.~\ref{fig:comm}. In contrast to the purely diagonal terms mentioned above, we are now constraint in choosing the signs so that these off-diagonal terms have non-positive entries. As a consequence of Eq.~(\ref{eq:cond1}) one can choose $\Delta_{YY}=-{\rm sign}(\alpha^{YY})$, $\Delta_{XX}=-{\rm sign}(\alpha^{XX})$, individually for each edge between qubits $2i$ and $2i+1$, so all X-like terms are stoquastic. Note that in case a term which is mapped to a X-like term occurs at the boundary we replace its action on non-existing qubits by $I$.
\end{proof}

{\em Remark}: For a disordered Heisenberg chain with periodic boundary this construction would not work as the XX and YY terms do not form an independent set.

\begin{acknowledgments}
	We thank Itay Hen for interesting discussions. BMT and JK acknowledge funding from ERC grant EQEC No. 682726. MI acknowledges support by the DFG (CRC183, EI 519/14-1) and EU
    FET Flagship project PASQuanS. MM acknowledges support by the NSF under Grant No. CCF-1954960 and by IARPA and DARPA via the U.S. Army Research Office contract  W911NF-17-C-0050.

\end{acknowledgments}

\bibliographystyle{plainnat}
\bibliography{apssamp2.1}

\appendix


\section{Stoquastic Hamiltonians and the sign problem in the path integral quantum Monte Carlo method}
\label{sec:signProblem}

In this Appendix we prove that sparse Hamiltonians $H$ which are globally stoquastic by an efficiently-computable curing transformation avoid the sign problem in the path integral Monte Carlo method.
As target for the path integral Monte Carlo method we focus on estimating $\langle H \rangle_{\beta}={\rm Tr} H \rho_{\beta}$ where $\rho_{\beta}=\exp(-\beta H)/Z(\beta)$ is the Gibbs state.
We will argue that {\em if} there is Metropolis-algorithm based Markov chain which is rapidly mixing and efficient both in $n$ and $\beta$ then a high-accuracy ground state energy estimate can be obtained. Since estimating the ground state energy of such Hamiltonians with this precision is at least \textbf{StoqMA}-hard \cite{BBT:stoq}, --the class includes frustrated classical Hamiltonians for which determining the ground state energy with this accuracy is \NP{}-complete--, it is clear that the Metropolis method cannot always be rapidly mixing. In fact, as is well-known, the configuration updates in the Metropolis algorithm can be chosen in various ways which can lead to better or worse convergence.
Indeed, even when there is an efficient quantum adiabatic algorithm using a stoquastic Hamiltonian, the Monte Carlo method can fail \cite{HF:obstructions}. 

Hence the Theorem below captures the fact that determining a rapidly-mixing Metropolis Markov chain with a good starting point which has sufficient overlap with the fixed point probability distribution is the bottle-neck in such path integral quantum Monte Carlo methods. Ref.~\cite{Crosson2018} proves that for some 1D stoquastic models one can set up such rapidly-mixing Metropolis Markov chain, although the running time for these algorithms scales as ${\rm poly}(e^{\beta})$ (and not ${\rm poly}(\beta)$).

Our proof takes some inspiration from \cite{sandvik:lect}. An alternative approach for estimating the expectation value of observables based on an estimation of $Z(\beta)={\rm Tr} e^{-\beta H}$ itself at various $\beta$ has been described in \cite{Crosson2018}. 

In Section \ref{sec:sign} we further discuss the manifestation of the sign problem when the Hamiltonian is not stoquastic.

\begin{theorem}[Loosely]
	If there exists an efficient Metropolis Monte Carlo Markov chain for the path integral Monte Carlo estimate of $\langle H\rangle_{\beta}$ which converges efficiently in ${\rm poly}(n)$ and ${\rm poly}(\beta)$ for all $\beta$, then this would provide an efficient algorithm to estimate the ground state energy of $H$ with $1/{\rm poly}(n)$ accuracy.
	\label{thm:sign}
\end{theorem}

We first provide a small proposition which relies on previous observations in \cite{BG:poly}:

\begin{proposition}
	Assume a probabilistic ${\rm poly}(n)$-time classical algorithm which outputs an estimate $H^{\rm est}_{\beta}$ for $\langle H \rangle_{\beta}$ such that 
	\begin{equation}
	{\rm Prob}(|H^{\rm est}_{\beta}-\langle H \rangle_{\beta} | < \varepsilon) \geq 1-\delta,
	\end{equation}
	for some constant $\delta < 1$ such that the algorithm has running time ${\rm poly}(n,1/\varepsilon,\beta)$. Then there exists a ${\rm poly}(n)$ time algorithm which can estimate the ground state energy $E_0$ of $H$ as $E^{\rm est}_0$ with $1/{\rm poly}(n)$ accuracy.
	\label{prop:free}
\end{proposition}

\begin{proof}
	The free energy is given by 
	\begin{align*}
	    	F(\beta)&=-\frac{1}{\beta} \log Z \\&=\langle H\rangle_{\beta}-\beta^{-1} S(\rho_{\beta}),
	\end{align*}

	with $S(\rho)=-{\rm Tr} \rho \ln \rho$.	We have $|F(\beta)-E_0| \leq \beta^{-1} n$ for a system of $n$ qubits, see \cite{BG:poly}. Hence $|\langle H \rangle_{\beta}-E_0| \leq 2 \beta^{-1} n\leq 1/\text{\rm poly}(n)$ for $\beta={\rm poly}(n)$. Hence by choosing $\beta$ polynomially large in $n$, the classical algorithm which outputs the estimate $H^{\rm est}_{\beta}$ gives, with probability larger than $1-\delta$, a $1/{\rm poly}(n)$ approximation to $E_0$.
\end{proof}

Another standard tool we use is 

\begin{proposition}[Chernoff-Hoeffding Inequality]
	Let $X_1, \ldots, X_K$ be random i.i.d. variables such that $a \leq X_i \leq b$ and let $X=\frac{1}{K}\sum_{i=1}^K X_i$ and $\mu=\mathbb{E}(X)$. Then
	\begin{align}
	\mathbb{P}((1-\epsilon) \mu \leq X \leq (1+\epsilon) \mu)&\geq 1-2 e^{-\frac{2 K \epsilon^2 \mu^2}{(b-a)^2}}.
	\label{eq:ch}
	\end{align}
	\label{prop:chern}
\end{proposition}


{\em Proof of Theorem \ref{thm:sign}}:
Let $H$ be the sparse $n$-qubit Hamiltonian which is globally stoquastic in the standard computational basis $\{\ket{x}\}_{x=0}^{2^n-1}$. Thus, upon being given $\ket{x}$, there are ${\rm poly}(n)$ $y$ such that $\bra{y} H \ket{x}\neq 0$ and these can be efficiently computed. A path integral representation of $\langle H \rangle_{\beta}$ is given by 
\begin{align}
\langle H \rangle_{\beta} &=\sum_{\mathbf{x}}h(\mathbf{x})P(\mathbf{x},\beta)+ \epsilon_{\rm trot}(N,\beta, ||H||)
\label{eq:trot}
\end{align}
with local energy 
\begin{align}
h(\mathbf{x})&=\frac{1}{N-1}\sum_{i=1}^{N}\frac{\bra{x_{i+1}}H(I-\tau H)\ket{x_{i}}}{\bra{x_{i+1}}I-\tau H \ket{x_{i}}}
\label{eq:local-energy}
\end{align}
which depends on the path $\mathbf{x}=(x_1,...x_N)$, with $n$-bit configuration $x_i$. Here
$ \tau=\beta/N$ with $N$ the Trotter step and the probability distribution $P(\mathbf{x}, \beta)$ is given by
\begin{align}
P(\mathbf{x}, \beta)&=\frac{W(\mathbf{x}, \beta)}{
	\sum_{\mathbf{x}} W(\mathbf{x}, \beta)} \label{3}, 
\end{align} 
with 
\begin{align}
W(\mathbf{x}, \beta)=\prod_{i=1}^{N} \bra{x_{i+1}}I-\tau H\ket{x_{i}}\delta_{x_1,x_{N+1}}.
\label{eq:def-weight}
\end{align}
We observe that $P(\mathbf{x}, \beta)$ is a probability distribution when $H$ is globally stoquastic as for all ${\bf x}, \beta$ the weight $W(\mathbf{x},\beta)\geq 0$.
The Trotter step $N$ can be chosen to be $N={\rm poly}(n)$ such that the Trotter error $\epsilon_{\rm trot}(N,\beta ,||H||)$ is $1/{\rm poly}(n)$ small. \footnote{We gloss over the standard analysis of the Trotter error here.}

The goal of the Metropolis algorithm is to set up a Markov chain which allows one to sample from the probability distribution $P({\bf x},\beta)$. The formulation of such Metropolis algorithm for some fixed $\beta, N$ requires that one can efficiently (in $n$) calculate the ratio $\frac{P({\bf x},\beta)}{P({\bf y},\beta)}$ for any two given ${\bf x}$ and ${\bf y}$. This is possible as the weight $W({\bf x,\beta})$ for a path ${\bf x}$ can be efficiently computed using the ability to compute $\bra{x} H \ket{y}$.

Now we assume that it is possible to set up a Metropolis algorithm which allows one to generate samples ${\bf x}_1, \ldots, {\bf x}_K$ from $P(\mathbf{x},\beta)$ efficiently. For each such sample, one computes $h({\bf x})$ and to prove how the sample mean $\frac{1}{K}\sum_{k=1}^K h({\bf x}_k)$
deviates from $\sum_{\mathbf{x}}h(\mathbf{x})P(\mathbf{x},\beta)$, we need to upperbound $|h({\bf x})|$, as it allows us to apply the Chernoff bound in Proposition \ref{prop:chern}. It is important to observe that we need this upper bound to be some ${\rm poly}(n)$ again, so that by choosing large enough $K={\rm poly}(n)$ one can suppress the error $\epsilon \mu$.
We observe that a path ${\bf x}$ has zero probability with respect to $P(\mathbf{x}, \beta)$ when it contains a segment $x_i, x_{i+1}$ with $\bra{x_{i+1}} I-\tau H \ket{x_i}=0$. Assume that we accept a new configuration in our Metropolis algorithm only when this new configuration ${\bf x}$ has the property that $\forall i\colon, \bra{x_{i+1}} I-\tau H \ket{x_i} \neq 0$: this property is easy to verify namely one checks whether $x_{i+1}=x_i$ or, if not, $\bra{x_{i+1}} H \ket{x_i} \neq 0$, only these paths are allowed.  In addition, disallowing these paths does not alter what is the fixed point of the Markov chain, namely $P({\bf x}, \beta)$.

Now we bound $|h({\bf x})|$ in Eq.~(\ref{eq:local-energy}) for all such allowed path configurations, i.e. 
\begin{align}
&\left|\frac{\bra{x_{i}}H(I-\tau H)\ket{x_{i}}}{\bra{x_{i}}I-\tau H \ket{x_{i}}}\right| \leq
\frac{H_{x_i,x_i}- \tau \sum_{y}H_{y,x_i}^2}{1- \tau H_{x_i,x_i}} \leq \nonumber\\
&|H_{x_i,x_i}|(1+2\tau |H_{x_i,x_i}|)  \nonumber\\
&x_i \neq x_{i+1}\colon \left|\frac{\bra{x_{i+1}}H(I-\tau H)\ket{x_i}}{\bra{x_{i+1}}I-\tau H \ket{x_i}}\right| =\left|\frac{1}{\tau}-\frac{H^2_{x_{i+1},x_i}}{H_{x_{i+1},x_i}}\right| \nonumber \nonumber\\ &\leq  
\frac{1}{\tau}+\frac{|H^2_{x_{i+1},x_i}|}{|H_{x_{i+1},x_i}|}.
\label{eq:bounds}
\end{align}
We see that when $||H_{x,y}||=O(1)$ and $\tau={\rm poly}(n)$, this results in the average value $h({\bf x})$ being polynomially bounded in $n$. 

The total error $\varepsilon$ in $H_{\beta}^{\rm est}$ versus $\langle H\rangle_{\beta}$ in Proposition \ref{prop:free} is both due to the $1/{\rm poly}(n)$ Trotter error in Eq.~(\ref{eq:trot}) and the finite-$K$ sampling error. Since $\mu=\sum_{\bf x} P(\mathbf{x},\beta) h({\bf x})\leq {\rm poly}(n)$, $K={\rm poly}(n)$ can chosen in the application of the Chernoff-Hoeffding inequality in Eq.~(\ref{eq:ch}), so that the error $\epsilon \mu=1/{\rm poly}(n)$ and $\delta$ is exponentially small. Hence, collecting such $1/{\rm poly}(n)$ errors would lead to the existence of a probabilistic polynomial time path integral Monte Carlo algorithm which satisfies the assumption in Proposition \ref{prop:free}.\\


\subsection{Sign Problem}
\label{sec:sign}

When applied to general (real) Hamiltonians the path integral Monte Carlo algorithm suffers from the so-called sign problem as for general real Hamiltonians the weights $W(\mathbf{x},\beta)$ in Eq.~(\ref{eq:def-weight}) can also be negative, making $P(\mathbf{x}, \beta)$ a quasi-probability distribution. 
A standard method is then to define a genuine probability distribution 
\begin{equation}
\tilde{P}(\mathbf{x},\beta)=\frac{|W(\mathbf{x}, \beta)|}{\sum_{\mathbf{x}} |W(\mathbf{x}, \beta)|},
\end{equation}
such that 
\begin{equation}
\langle H \rangle_{\beta} \approx \frac{\sum_{\bf x} [h({\bf x}){\rm sign}(P({\bf x},\beta))] \tilde{P}(\mathbf{x},\beta) }{\sum_{\mathbf{x}}{\rm sign}(W(\mathbf{x}, \beta)) \tilde{P}(\mathbf{x},\beta)}.
\label{eq:approxH-sign}
\end{equation}

One can see the probability distribution $\tilde{P}({\bf x},\beta)$ as the fixed point of some Metropolis Monte Carlo algorithm which uses a stoquastified or `designed' version of $H$ \cite{Hangleiter2019, CAHY:design} defined as $\tilde{H}$:
\begin{align}
\Tilde{H}_{xy}&=-|H_{xy}|, \ \ \  \forall x\neq y \nonumber \\
\tilde{H}_{xx}&=H_{xx} .
\label{def:design}
\end{align}
Note that sparsity or locality of $\tilde{H}$ and the efficient evaluation of $\bra{x} \tilde{H} \ket{y}$ hold when they hold for $H$, hence $\tilde{H}$ can be used to determine a Metrolis algorithm. Since
\begin{align*} 
|W({\bf x}, \beta)| & = \prod_{i=1}^N \left \vert \bra{x_{i+1}}I-\tau \tilde{H}\ket{x_{i}}\right \vert \delta_{x_1, x_{N+1}}, 
\end{align*}
the Metropolis Markov chain towards the distribution $\tilde{P}({\bf x},\beta)$ can use $\tilde{H}$.

One can thus similarly apply a Metropolis Markov chain (whose convergence is not guaranteed as previously). However, {\em even if} such efficient and convergent Metropolis algorithm were to exist, it does not imply that one can estimate $\langle H \rangle_{\beta}$ with $1/{\rm poly(n)}$ accuracy as the denominator in Eq.~(\ref{eq:approxH-sign}) can be exponentially small in $n$. In order to estimate $\langle H \rangle_{\beta}$ accurately, one needs to make the relative error of the denominator $\epsilon=1/{\rm poly}(n)$. When the mean is exponentially small, the Chernoff-Hoeffding inequality in Proposition \ref{prop:chern} says that this would take an exponential number of samples $K$, which is the crux of the sign problem.

\subsection{Translationally-invariant 1D XYZ Model}
\label{app:XYZ}

The translationally invariant 1D XYZ model,
\begin{equation}\label{translinvar}
H=\sum_{i}\alpha X_iX_{i+1}+\beta Y_iY_{i+1}+\gamma Z_iZ_{i+1},
\end{equation}
is known to be always sign-problem free, even when the Hamiltonian itself is not stoquastic in the standard basis \cite{Suzuki}. The following argument holds for open and closed boundaries.

Here we focus on evaluating $Z(\beta)={\rm Tr} \exp(-\beta H)=\sum_{\bf x} W({\bf x,\beta})+\epsilon_{\rm trot}$ and imagine using the method in \cite{Crosson2018} to determine the expectation of observables in the Gibbs state. From Eq.(~\ref{3}) we can see $ W({\bf x,\beta})$ as a product of weights on a closed path starting and ending at $x_1$.  One can argue that every such path ${\bf x}$ has non-negative weight $W({\bf x},\beta)$ as follows. Every path is constructed from some $m$ insertions of $-\tau (\alpha X_iX_{i+1}+\beta Y_iY_{i+1})$ for some $i$'s and diagonal factors. For the path to come back to the initial string, it certainly needs to be true that $m$ is even. Since
\begin{align*}
&\alpha X_iX_{i+1}+\beta Y_iY_{i+1} \ket{0_i0_{i+1}}=(\alpha-\beta) \ket{1_i1_{i+1}} \\
&\alpha X_iX_{i+1}+\beta Y_iY_{i+1} \ket{1_i1_{i+1}}=(\alpha-\beta) \ket{0_i0_{i+1}} \\
&\alpha X_iX_{i+1}+\beta Y_iY_{i+1} \ket{1_i0_{i+1}}=(\alpha+\beta) \ket{0_i1_{i+1}} \\
&\alpha X_iX_{i+1}+\beta Y_iY_{i+1} \ket{0_i1_{i+1}}=(\alpha+\beta) \ket{1_i0_{i+1}},
\end{align*}
it also holds that the number $m_1$ of insertions to a $\ket{00}$ or $\ket{11}$ bitstring segment must be even, since the insertion to a bitstring segment $\ket{10}$ or $\ket{01}$ merely moves the position of the bits. So $m_2=m-m_1$, the number of insertions to segments $\ket{10}$ or $\ket{01}$ must also be even. Due to translational invariance, the weight of a path with $m=m_1+m_2$ insertions is 
$(-\tau)^m(\alpha+\beta)^{m_2} (\alpha-\beta)^{m_1}>0$, hence we can replace $W({\bf x}, \beta)$
by $|W({\bf x},\beta)|$ without modifying $Z(\beta)$.
When the chain is not translationally-invariant, this simple argument no longer applies. 

To argue that there are indeed choices for $\alpha, \gamma$ and $\delta$ such that $H$ in Eq.~(\ref{translinvar}) is not stoquastic by single-qubit (product) transformations or even by Clifford transformations despite this representation of $Z(\beta)$, let's take a closed chain with an odd number of edges. For such XYZ models single-qubit unitaries are as powerful as single-qubit Cliffords \cite{Klassen2019}. Single-qubit Cliffords can only permute Paulis and add signs, so we permute the Paulis on each qubit so that $|\alpha|$ is largest and hence $|\alpha| \geq |\beta|$. If $\alpha < 0$, then we are done. If however, $\alpha > 0$, we cannot apply Paulis such that $\alpha < 0$ for all edges. Thus in this case the XYZ model is not stoquastic by single-qubit Clifford transformations (nor by the transformations in Theorem \ref{clifford2}).

\section{Single-Qubit Clifford Transformations Fall Short For Some Cases}
\label{sec:fs}

In \cite{KT:stoq} and \cite{Klassen2019} an efficient algorithm was proposed for removing, if possible, the sign problem under single-qubit product unitaries. Applying this to a simple 1D XYZ model we will show that the sign problem cannot always be cured by single-qubit unitary transformations. 

For the Heisenberg XYZ model on a 1D chain in Eq.~(\ref{eq:heis}), let the matrices $\beta_{i\ i+1}={\rm diag}(\alpha_{i\ i+1}^{XX},\alpha_{i \ i+1}^{YY},\alpha_{i \ i+1}^{ZZ})$. As has been shown in \cite{KT:stoq}, $H$ can be sign-cured if there exist $3 \times 3$ signed permutation matrices $\Pi_i$, with unit determinant (corresponding to single-qubit Cliffords where the signs are due to the Paulis) such that
\begin{equation*}
\Pi_i\beta_{i\ i+1}\Pi_{i+1}=\Tilde{\beta}_{i\ i+1},
\end{equation*}
where $\Tilde{\beta}_{i\ i+1}=\text{diag} (\Tilde{\alpha}_{i \ i+1}^{XX},\Tilde{\alpha}_{i\ i+1}^{YY},\Tilde{\alpha}_{i\ i+1}^{ZZ})$ and $\Tilde{\alpha}_{i\ i+1}^{XX}\leq-|\Tilde{\alpha}_{i\ i+1}^{YY}|$ for all $i$. When $\beta_{i\ i+1}$ and $\beta_{i+1 \ i+2}$ are both of rank larger than or equal to two, the Clifford transformation inducing the permutations needs to be the same on qubits $i,i+1$ and $i+2$ to preserve the diagonal form of the $\beta$-matrices \cite{KT:stoq}. This condition allows us to construct examples of Hamiltonians which cannot be sign-cured using single-qubit unitaries. 

As a simple example consider the following three qubit Hamiltonian 
\begin{align*}
H_{123}=2X_1X_2+Z_1Z_2+X_2X_3+3Y_2Y_3+2Z_2Z_3,
\end{align*}
which corresponds to $\beta_{12}=\text{diag}(2,0,1),\ \beta_{23}=\text{diag}(1,3,2)$. We see that the YY coefficient is the largest in absolute value of the coefficients in $\beta_{23}$ and in particular $|\alpha_{23}^{YY}|>|\alpha_{23}^{XX}|$. We can convince ourselves that there does not exist a signed permutation under which $\Tilde{\alpha}_{23}^{XX}\leq-|\Tilde{\alpha}_{23}^{YY}|$, $\Tilde{\alpha}_{12}^{XX}\leq-|\Tilde{\alpha}_{12}^{YY}|$. Indeed, the same argument holds for every Hamiltonian which contains $\beta_{i \ j}=\text{diag}(a,0,b)$ with $a>b$ and $\beta_{j \ k}=\text{diag}(c,d,e)$ with $c<e<d$ and $a,b,c,d,e \neq 0$. We have thus established that single-qubit unitaries are not always sufficient to sign-cure 1D XYZ Hamiltonians.

\end{document}